\documentclass[12pt]{iopart}

\usepackage{iopams}  
\usepackage{dcolumn}% Align table columns on decimal point
\usepackage{bm}% bold math
\usepackage{amssymb,amsthm}
\usepackage{graphicx}
\usepackage{color,setstack}
\usepackage{enumerate}
\newcommand\bG{\mathbf{G}}

\newcommand\EE{{\mathbb E}}

\newcommand\RR{{\mathbb R}}
\newcommand\TT{{\mathbb T}}
\newcommand\ZZ{{\mathbb Z}}

\newcommand{\ri}{\mathrm{i}}

\newcommand{\mc}[1]{{\mathcal #1}}
\newcommand{\mf}[1]{{\mathfrak #1}}
\newcommand{\mb}[1]{{\mathbf #1}}
\newcommand{\bb}[1]{{\mathbb #1}}

\newtheorem{prop}{Proposition}
\newtheorem{theo}{Theorem}
\newtheorem{lemma}{Lemma}

\newtheorem{cor}{Corollary}
\newtheorem{definition}{Definition}
\newtheorem{assumption}{Assumption}

\def\sqw{\hbox{\rlap{\leavevmode\raise.3ex\hbox{$\sqcap$}}$%
\sqcup$}}
\def\cqfd{\ifmmode\sqw\else{\ifhmode\unskip\fi\nobreak\hfil
\penalty50\hskip1em\null\nobreak\hfil\sqw
\parfillskip=0pt\finalhyphendemerits=0\endgraf}\fi}

\begin{document}

\title[]{Anomalous diffusion for a class of systems with two conserved quantities}

\author{C\'edric Bernardin}
\address{
Universit\'e de Lyon and CNRS, UMPA, UMR-CNRS 5669, ENS-Lyon,
46, all\'ee d'Italie, 69364 Lyon Cedex 07 - France.
}
\ead{Cedric.Bernardin@umpa.ens-lyon.fr}

\author{Gabriel Stoltz}
\address{
Universit\'e Paris Est, CERMICS and INRIA, MICMAC project-team, 
Ecole des Ponts ParisTech, 6 \& 8 Av. Pascal, 77455 Marne-la-Vall\'ee, France
}
\ead{stoltz@cermics.enpc.fr}

\date{\today}

%dans l'ordre : Classical dynamic and nonequilibrium statistical mechanics (general), Transport processes, Stochastic methods (Fokker-Planck, Langevin, etc.) 

\begin{abstract}   
We introduce a class of one dimensional deterministic models of energy-volume conserving 
interfaces. Numerical simulations show that these dynamics are genuinely super-diffusive. 
We then modify the dynamics by adding a conservative stochastic noise so that it becomes 
ergodic. System of conservation laws are derived as hydrodynamic limits of the modified dynamics. 
Numerical evidence shows these models are still super-diffusive. This is proven 
rigorously for harmonic potentials. 
\end{abstract}

\submitto{\NL}

\maketitle

%------ INTRODUCTION ----------------
\section{Introduction}

Over the last decade, transport properties of one-dimensional systems consisting 
of coupled oscillators on a lattice have been the subject of many theoretical and 
numerical studies, see the review papers~\cite{BLR,D,LLP}. 
Despite many efforts, our knowledge of the fundamental mechanisms necessary and/or sufficient 
to have a normal diffusion remains very limited. 
Nevertheless, it has been recognized that conservation of momentum plays a major 
role and numerical simulations provide a strong evidence of the fact that 
one dimensional chains of anharmonic oscillators conserving momentum are super-diffusive. 

In this paper we propose a new class of models for which anomalous diffusion is observed. 
The system under investigation presents several analogies with standard 
chains of oscillators, but, and it is our main motivation, has a simpler mathematical structure. 

\subsection{Description of the model}
\label{sec:description_model}

Let $U$ and $V$ be two potentials on $\RR$ and consider the Hamiltonian system $(\omega (t) )_{t \ge 0} = ( \, {\bm r} (t) , {\bm p} (t) \,)_{t \ge 0}$ described by the equations of motion
\begin{equation}
\label{eq:generaldynamics}
\frac{dp_x}{dt} = V'(r_{x+1}) -V'(r_x), 
\qquad \frac{dr_x}{dt} = U' (p_x) -U' (p_{x+1}), 
\qquad x \in \ZZ, 
\end{equation}
where $p_x$ is the momentum of particle $x$, $q_x$ its position and $r_x=q_{x} -q_{x-1}$ the ``deformation''. Standard chains of oscillators are recovered for a quadratic kinetic energy $U(p)=p^2 /2$. The dynamics conserves three physical quantities: the total momentum $\sum_{x} p_{x}$, the total deformation $ \sum_{x} r_{x}$ and the total energy $\sum_x {\mc E_x}$ with ${\mc E}_x= V(r_x) + U(p_x)$. Consequently, every product probability measure ${\nu}_{\beta,\lambda, \lambda'}$ defined by
\begin{equation}
\label{eq:invmeas007}
d{\nu}_{\beta,\lambda,\lambda'} (\eta) = \prod_{x \in \ZZ} {\mc Z} (\beta,\lambda,\lambda')^{-1} 
\exp\left\{ -\beta {\mc E}_x -\lambda p_{x} -\lambda' r_{x} \right\} \, d r_x \, d p_x 
\end{equation}
is invariant under the evolution. For later purposes, let us also 
introduce the short-hand notation
\begin{equation*}
{\varepsilon}_x= \left(
\begin{array}{c}
{\mc E}_x \\
p_{x}\\
r_{x}
\end{array}
\right), 
\qquad {u}= 
\left(
\begin{array}{c}
\beta\\
\lambda\\
\lambda'
\end{array}
\right), 
\qquad {\bar \varepsilon} =
\left(
\begin{array}{c}
\nu_{\beta,\lambda,\lambda'} ( {\mc E}_0 )\\
{\nu_{\beta,\lambda,\lambda'}} (p_0)\\
{\nu_{\beta,\lambda,\lambda'}} (r_0)
\end{array}
\right).
\end{equation*} 

In this paper, we are mainly interested in the case $U=V$ (except for instance in Section~\ref{sec:secgeneral}), which has the advantage of introducing more symmetries into the problem. The state of the system at time~$t$ is then more conveniently described by the variable
$\eta (t) =\left\{ \eta_{x} (t); x \in \ZZ \right\} \in \RR^{\ZZ}$ with $\eta_{2x-1}=r_x$ and $\eta_{2x}=p_x$. The dynamics can be rewritten as:
\begin{equation}
\label{eq:dyneq}
d\eta_{x} (t) =\Big(V' (\eta_{x+1}) - V' (\eta_{x-1})\Big) dt.
\end{equation}
The system can therefore also be interpreted as a fluctuating interface where the algebraic volume of the interface at site $x \in \ZZ$ is given by $\eta_x$ and the energy by $V(\eta_x)$. We focus our study on the (anomalous) diffusion of the energy. The three quantities, momentum $\sum_{x} \eta_{2x}$, deformation $ \sum_{x\in\mathbb{Z}} \eta_{2x+1}$ and energy $\sum_{x\in\mathbb{Z}} V(\eta_x)$,  are conserved but, for reasons which will become clear later, our interest lies only in the ``volume'' $\sum_{x\in\mathbb{Z}} \eta_x$ and in the energy.

\subsection{Hydrodynamic limit}

Energy transport properties depend strongly on the chosen time-scale. 
The first natural scale to consider is the hyperbolic scale where the system is followed
on long times $N t$, the space being renormalized by a factor~$N^{-1}$, 
with $N \to \infty$.

To this end, we define the empirical energy/momentum/deformation measure as follows:
\begin{equation}
\pi^N (t, dq) = N^{-1} \sum_{x \in \ZZ} \varepsilon_{x} (tN) \, \delta_{x/N} (dq), 
\qquad  q \in \RR.
\end{equation}
%\begin{equation}
%  \label{eq:empirical_measure}
%  \pi^N (t, dq) = N^{-1} \sum_{x \in \ZZ} \xi_{x} (tN) \, \delta_{x/N} (dq), 
%  \qquad 
%  q \in \RR.
%\end{equation}
At time $t=0$, this measure is supposed to converge in probability to some 
macroscopic profile $\Pi_0 (q) \, dq$ which has a density w.r.t. the Lebesgue measure. 
If we assume that a \textit{local equilibrium hypothesis} holds, 
it is not difficult to show that the expected macroscopic evolution 
equation for $\Pi(t,q) \, dq =\lim_{N \to \infty} \pi^N (t,dq)$ 
is given by a triplet of compressible Euler equations:
\begin{equation}
\label{eq:syslim0}
\partial_t \Pi + \partial_q {I}(\Pi) =0, 
\qquad 
\Pi (0,\cdot)=\Pi_0 (\cdot),
\end{equation} 
where ${I} (\Pi) \in \RR^3$ is the macroscopic current whose explicit expression is not important here. Unfortunately, proving the local equilibrium hypothesis for (\ref{eq:generaldynamics}) is out of the range of the current mathematical techniques. The main difficulty comes from our inability to show that the dynamics is ergodic (in the sense of Definition~\ref{def:ergo} below). If the ergodicity is proved, in the time interval where (\ref{eq:syslim0}) has a smooth solution, the relative entropy method of Yau (see~\cite{Y}) can be adapted to show that the system has~(\ref{eq:syslim0}) as a hydrodynamic limit. Hence, the problem can be reduced to proving the ergodic behavior of the underlying dynamics. Deriving the convergence to~(\ref{eq:syslim0}) after the shocks is considerably more difficult since even the concept of a solution to a system of conservation laws is not fully understood.
 
To overcome the above mentioned lack of ergodicity of deterministic systems, 
it has been proposed to add a stochastic perturbation to the dynamics. The theory of stochastic perturbations of Hamiltonian dynamics has a long history. To our knowledge, the first paper on the ergodicity of infinite lattice models is~\cite{F00} (see also ~\cite{BBO1,BBO2,FFL, FLO,GV,LO,OVY}). The added noise must be carefully chosen in order not to destroy the conservation laws we are interested in. In the general case $U \ne V$, the Hamiltonian dynamics can be perturbed by a local noise acting on the velocities (as proposed in \cite{FFL}) but conserving the three physical invariants mentioned at the end of Section~\ref{sec:description_model} (see Section \ref{sec:secgeneral}). With such additional noises, the perturbed dynamics can be proved to be ergodic (see Theorem \ref{th:thgeneral}), so that~(\ref{eq:syslim0}) is obtained as a hydrodynamic limit. However, our motivation being to simplify as much as possible the dynamics considered in~\cite{BBO1,BBO2} without destroyin
 g the anomalous behavior of the energy diffusion, we mainly focus on the symmetric case $U=V$ with a noise conserving only the two important quantities responsible of the anomalous transport behavior, namely the the energy and the volume. Thus, we introduce a new stochastic energy-volume  conserving dynamics, which is still described by (\ref{eq:dyneq}) between random exponential times where two nearest neighbors heights $\eta_x$ and $\eta_{x+1}$ are exchanged (see Subsection~\ref{subsec:sto} for a precise definition). 
Observe that the noise still conserves the total energy and the total volume but destroys the conservation of momentum and deformation. Therefore, only two quantities are conserved and the invariant Gibbs measures of the stochastic dynamics correspond to the choice $\lambda=\lambda'$ in~(\ref{eq:invmeas007}). We denote ${\nu}_{\beta,\lambda,\lambda}$ (resp. ${\mc Z}(\beta,\lambda,\lambda)$) by $\mu_{\beta, \lambda}$ (resp. $Z (\beta,\lambda)$) and we use in the sequel the short-hand notation
\begin{equation}
\label{eq:short_hand_notation}
\xi_x= \left(
\begin{array}{c}
V(\eta_x)\\
\eta_x
\end{array}
\right), 
\qquad {w}= 
\left(
\begin{array}{c}
\beta\\
\lambda
\end{array}
\right), 
\qquad {\bar \xi} =
\left(
\begin{array}{c}
\mu_{\beta,\lambda} (V(\eta_0))\\
{\mu_{\beta,\lambda}} (\eta_0)
\end{array}
\right).
\end{equation}

The first main result of this paper is that the perturbed dynamics informally described above 
is ergodic (see Theorem~\ref{th:1}). Consequently, before the appearance of shocks, 
the stochastic energy-volume conserving dynamics has a hyperbolic system of two conservation laws as a hydrodynamic limit (see Theorem~\ref{th:hl}).

\subsection{(Super)Diffusive limit}

In the second part of the paper we investigate the diffusion of the energy at a 
longer time-scale. If the process has a diffusive behavior then the relevant time scale 
is the diffusive one, where the system is studied over long times $N^2 t$ with 
space renormalized by a factor~$N^{-1}$. We claim that the system genuinely 
displays an anomalous energy diffusion so that the diffusive scale is not the 
relevant one. Heuristically, we can interpret this anomalous diffusion as 
a consequence of the volume conservation law (see Section~\ref{subsec:linmodel}).

We start the infinite system (\ref{eq:generaldynamics}) under the equilibrium distribution $\mu_{\beta, \lambda,\lambda'}$ 
and consider first the fluctuation field in the hyperbolic time scaling:
\begin{equation}
  \label{eq:YY}
  \mathcal{Y}_N (t,\bG) =\frac{1}{\sqrt{N}} \sum_{x\in \ZZ} 
  \bG\left(\frac{x}{N}\right) \otimes \left(\varepsilon_x (tN)  - {\bar \varepsilon}\right),
\end{equation} 
where
\[
\bG(y) = \left(
\begin{array}{c}
G_{1} (y)\\
G_2 (y)\\
G_3 (y)
\end{array} 
\right)
\]
is a smooth vector valued test function with compact support. 
We expect that ${\mc Y}_N (t, \cdot)$ converges in law to ${\mc Y} (t,\cdot)$, 
where ${\mc Y}$ is solution of the linearized equation
\begin{equation}
\label{eq:lf00}
\partial_t {\mc Y} + D{I}({\bar \varepsilon})  \, \partial_q {\mc Y}=0,
\end{equation}
with $D{I} ({\bar \varepsilon})$ the differential of $I$ at ${\bar \varepsilon}$. Hence, in the hyperbolic scaling, fluctuations evolve deterministically according to
\[
{\mc Y} (t,\bG) ={\mc Y} \left(0, e^{t \, {\mc U}^*} \bG\right),
\]
where ${\mc U}= [D{I} ({\bar \varepsilon})] \partial_q$ and ${\mc U}^*= - [D{I} ({\bar \varepsilon})]^* \partial_q$  with $[D{I} ({\bar \varepsilon})]^*$ the transpose matrix of $D{I} ({\bar \varepsilon})$.

To see a nontrivial behavior of the fluctuation field, we need to look at
${\mc Y}_N$ on a longer time scale $t N^{1+\alpha}$, for some $\alpha>0$. It is expected that, after subtracting the transport term 
appearing in the hyperbolic time scale, the field
\begin{equation}
  \label{eq:696}
  {\widetilde{\mc Y}}_N (t,\bG) = {\mc Y}_N \left(t N^{\alpha}, e^{-t N^{1+\alpha}  {\mc U}^*} \bG\right)
\end{equation}
converges to some limiting field ${\widetilde {\mc Y}}$. The case $\alpha=1$ would correspond to a diffusive behavior with ${\widetilde {\mc Y}}$ the solution of the linear stochastic partial differential equation 
\begin{equation*}
\partial_t {\widetilde {\mc Y}}= \nabla \cdot \left( {\widetilde {\mc D}}^{\infty}\,  \nabla {\widetilde {\mc Y}}\right) + \sqrt{2 {\widetilde \chi} {\widetilde {\mc D}}^{\infty} }\,  \nabla \cdot W,   
\end{equation*}
where $W(x,t)$ is a standard space-time white noise. Here, ${\widetilde\chi}$ is the {\textit{compressibility}} and ${\widetilde{\mc D}}^{\infty} = \lim_{ t \to \infty} {\widetilde{\mc D}}_{\beta,\lambda,\lambda'} (t)$ the limiting {\textit{diffusivity}} (see (\ref{eq:chi0}) and (\ref{eq:diff222}) for the definitions of these quantities). We refer the reader  to \cite{Sp} for a general background reference, and to \cite{LOV} for a rigorous proof of the convergence for asymmetric simple exclusion processes.

For the class of models we consider, our conjecture is that $\alpha$ is in general strictly lower than $1$. The value of $\alpha$ and the nature of ${\widetilde {\mc Y}}$ are not expected to be universal and should depend on some specific properties of the potentials. We also expect a similar picture when the deterministic dynamics (\ref{eq:generaldynamics}) is replaced by the stochastic energy-volume conserving dynamics (we then denote by ${\mc D}_{\beta,\lambda}$ and $C_{\beta,\lambda}$ the corresponding diffusivity and current-current correlation function). This anomalous value of $\alpha$ should be reflected in the divergence of the diffusivity ${\widetilde{\mc D}}_{\beta,\lambda,\lambda'} (t)$ (${\mc D}_{\beta,\lambda} (t)$ for the stochastic dynamics) in the large time limit  $t \to \infty$.

We are not able to study theoretically this problem for the deterministic dynamics 
and we have to turn to computer simulations of nonequilibrium systems in their steady-states. 
A chain of length $2N+1$ is coupled at each extremity (left and right) to a thermal reservoir fixing the temperature ($T_\ell$ on the left, $T_r$ on the right). In the stationary state $\langle \cdot \rangle_{\rm ss}$, the average current $\langle J_N \rangle_{\rm ss}$ is measured (see Section~\ref{sec:sim} for more precise definitions). 
The quantity of interest is the divergence exponent $\delta$ of the transport coefficient 
\[
\kappa_N = \frac{\langle N J_N \rangle_{\rm ss}}{T_\ell -T_r} \asymp N^{\delta}. 
\]
Anomalous diffusion corresponds to $\delta>0$.  

For a normal transport, the link between the two situations can be seen through a Green-Kubo formula for the limiting diffusivity, which expresses the latter as a quantity proportional to the time integral of the equilibrium current-current correlation function (see~(\ref{eq:cbl}) and~(\ref{eq:diff}) for a precise definition). It is widely accepted, but not proved, that for normal diffusive systems the limiting diffusivity coincides with the transport coefficient $\kappa =\lim_{N \to \infty} \kappa_N$.  For anomalous diffusion transport, $\kappa =+ \infty$ and the current-current autocorrelation function which appears in the Green-Kubo formula is not integrable because it decays too slowly. Consequently the limiting diffusivity is infinite.

Our second main results are the following. First, we show numerically that, for generic anharmonic potentials~$V$, the dynamics (\ref{eq:dyneq}) has an anomalous diffusion. We also show that this phenomenon persists if the deterministic dynamics is replaced by the stochastic energy-volume conserving dynamics, and that the divergence exponent~$\delta$ depends on the strength of the random perturbation. Secondly, for the stochastic energy-volume conserving dynamics and a harmonic potential $V(r)=r^2/2$, we compute explicitly the equilibrium current-current correlation function $C_{\beta,\lambda} (t)$ and show that it decays as $t^{-1/2}$ for large $t$ (see Theorem~\ref{th:GK}). This implies that the diffusivity ${\mc D}_{\beta, \lambda} (t)$ diverges as~$\sqrt{t}$, which is a clear manifestation of the super-diffusion of the energy for this model.

\subsection{Organization of the paper}
We  present more precisely the model under investigation in Section~\ref{sec:models}.
We first study the hydrodynamic limit in Sections~\ref{sec:hyp_scaling} (presentation
of the general result) and~\ref{sec:erg} (proof of the fundamental ingredient for the limit
to hold, namely the ergodicity of the dynamics).
We then consider diffusion properties, starting with analytical results on the longtime
tail of the current autocorrelation function, 
which can be obtained for harmonic potentials (see Section~\ref{sec:diffusive}),
and providing then scalings of the energy current obtained by 
numerical simulations of nonequilibrium systems
in their steady-states in Section~\ref{sec:sim}. 
Some proofs are gathered in the Appendix.

%------------- MODELS -------------------------------
\section{The models}
\label{sec:models}

\subsection{The deterministic models}

\subsubsection{Finite systems}
Consider the finite box $\Lambda_N=\{-N, \ldots,N\} \subset \ZZ$ 
(with $N \ge 1$). The product space $\RR^{\Lambda_N}$ is denoted by $\Omega_{N}$, 
and a typical element of $\Omega_N$ is $\eta=\{ \eta_x \in \RR\, ; \, x \in \Lambda_N\}$. 
The deterministic finite volume dynamics $(\eta^{N} (t))_{t \geq 0} \in \Omega_N$ is 
defined by its generator  
\begin{equation}
\label{eq:A}
\fl \qquad {\mathcal A}_{N} =\sum_{x =-(N-1)}^{N-1} 
\Big(V'(\eta_{x+1})-V'(\eta_{x-1})\Big)\partial_{\eta_x} -V' (\eta_{N -1}) \, \partial_{\eta_N} 
+ V' (\eta_{-N+1}) \, \partial_{\eta_{-N}},
\end{equation}
where $V$ is a smooth convex potential such that the partition function
\begin{equation*}
Z(\beta, \lambda) = \int_{-\infty}^{\infty} \exp\left( -\beta V(r) -\lambda r \right)\, dr
\end{equation*}
is well defined for $\beta>0$ and $\lambda \in \RR$.
The following microscopic energy-volume conservation laws hold for $x = -N+1,\dots,N-1$:
\begin{equation*}
{\mc A}_N V(\eta_x) = -\nabla \left[j^e_{x-1,x}\right], 
\qquad {\mc A}_N  \eta_x =-\nabla \left[ j^v_{x-1,x} \right], 
\end{equation*}
where $\nabla$ is the discrete gradient defined, for any function $u:\ZZ \to \RR$, 
by $(\nabla u) (x)=u(x+1)-u(x)$, and where the microscopic energy and volume 
currents are respectively
\begin{equation*}
j^e_{x,x+1} = -V'(\eta_{x+1}) V' (\eta_x), 
\qquad 
j^v_{x,x+1} =- \Big(V'(\eta_{x}) + V'(\eta_{x+1})\Big).
\end{equation*}

%\begin{remark}[Relationship with Hamiltonian systems]
%  \label{rmk:Ham}
%  Consider the variables $p_{x} = \eta_{2x-N}$ for $x \in \{0,1,\ldots,N\}$ 
%  and $r_x = q_{x}-q_{x-1} = \eta_{2x-1-N}$ for $x \in \{1,2,\ldots,N\}$.
%  Then, the equations of motion associated with the generator~(\ref{eq:A})
%  can be rewritten as 
%  \[
%  \eqalign{
%  dr_i = \Big(V'(p_i)-V'(p_{i-1})\Big)dt, \cr
%  dp_i = \Big(V'(q_{i+1}-q_i) - V'(q_i-q_{i-1})\Big)dt,
%  }
%  \]
%  except at the boundaries where $dp_0 = V'(q_{1}-q_0) \, dt$ and 
%  $dp_N = -V'(q_{N}-q_{N-1}) \, dt$.
%  They can be seen as the Hamiltonian evolution associated with 
%  \[
%  H(q,p) = \sum_{i=1}^N V(q_i-q_{i-1}) + \sum_{i=0}^N V(p_i),
%  \]
%  which describes a linear chain of oscillators with free boundary conditions
%  and a generalized (non-quadratic) kinetic energy.   
%\end{remark}

\subsubsection{Infinite systems.}
The dynamics in the infinite volume $\Lambda = \ZZ$, with formal generator
\[
\mathcal{A} = \sum_{x\in \ZZ} \Big(V'(\eta_{x+1})-V'(\eta_{x-1})\Big)\partial_{\eta_x}
\]
is also very important. Since the state space is unbounded, explosion problems can arise and the construction of solutions of the dynamics in infinite volume may become a technically non trivial problem. To avoid such issues, we restrict ourselves (apart from Section~\ref{sec:sim}) to the case 
$0 \leq V'' \le C$ for some positive constant $C>0$. 
Then, the construction is quite standard since the pioneering work \cite{LLL} of Lanford et al. We refer the interested reader to \cite{F1,FFL} for further precisions. 
For any $\alpha>0$, let $\Omega_{\alpha}$ be the set of configurations $\eta$ such that
\begin{equation*}
\sum_{x \in \ZZ} \eta_x^{2} \, e^{-\alpha |x|} < +\infty,
\end{equation*}
and equip $\Omega=\cap_{\alpha >0} \Omega_{\alpha}$ with its natural product topology 
and its Borel $\sigma$-field. The set of Borel probability measures on $\Omega$ is 
denoted by ${\mathcal P} (\Omega)$. A function $f: \Omega \to \RR$ is said to be \emph{local} 
if it depends of $\eta$ only through the coordinates $\{ \eta_x\, ; \, x \in \Lambda_f\}$, 
$\Lambda_f$ being a finite box of $\ZZ$. We also introduce the sets $C_0^k(\Omega)$ ($k \ge 1$) 
of bounded local functions on $\Omega$ which are differentiable 
up to order $k$ with bounded partial derivatives. 

For each initial condition $\sigma \in \Omega$ the existence and uniqueness of a solution 
to~(\ref{eq:dyneq}) can be proved by a classical fixed-point argument {\it \`a la} Picard. 
The solution $\eta(\cdot): =\eta(\cdot, \sigma)$ defines a process with continuous trajectories. Moreover each path $\eta (\cdot, \sigma)$ is a continuous and differentiable function of the initial data $\sigma$. We define the corresponding semigroup $(P_t)_{t \ge 0}$ by  $(P_t f)(\sigma) = f(\eta (t, \sigma))$ for any bounded measurable function $f$ on $\Omega$. 
The differentiability with respect to initial conditions shows that 
the Chapman-Kolmogorov equations 
\begin{equation}
\label{eq:CK}
\eqalign{
&(P_t f )(\sigma) = f(\sigma) +\int_0^t ({\mathcal A} P_s f)(\sigma) \, ds, 
\qquad 
f \in C_0^{1} (\Omega) \cr
&(P_t f )(\sigma) = f(\sigma) +\int_0^t (P_s {\mathcal A} f)(\sigma) \, ds, 
\qquad 
f \in C_0^{1} (\Omega).
}
\end{equation}
are valid.
With these equations, probability measures $\nu \in {\mathcal P} (\Omega)$ 
invariant by the deterministic dynamics are characterized by the stationary Kolmogorov equation:
\begin{equation*}
\forall f \in C_{0}^{1} (\Omega), 
\qquad 
\int ({\mathcal A} f)(\eta) d\nu (\eta) = 0.
\end{equation*}
Denoting the usual scalar product between two vectors $a,b \in \RR^2$ by $a \cdot b$,
and recalling the notation introduced in~(\ref{eq:short_hand_notation}),
it is easily seen that every product measure $\mu_{\beta,\lambda}$ defined by
\begin{equation*}
d\mu_{\beta,\lambda} (\eta) = \prod_{x \in \ZZ} Z(\beta,\lambda)^{-1} 
\exp\left\{ -w \cdot \xi_x \right\} \, d\eta_x
\end{equation*}
is invariant for the infinite dynamics. In the sequel, we 
denote the average of a function $f$ with 
respect to $\mu_{\beta,\lambda}$ by $\langle f \rangle_{\beta,\lambda}$. 

Depending on the potential $V$ at hand, the properties of the dynamics can be very 
different. In the next subsection we discuss the case $V(r)=r^2 /2$, which leads 
to a linear dynamics. A second remarkable potential is the exponential potential 
$V_{\rm KVM}(q)=e^{-q} +q -1$, corresponding to the so-called Kac-van-Moerbecke system, 
which is integrable (see~\cite{KVM}). The corresponding system is 
related to the famous Toda lattice~\cite{T}, 
\emph{i.e.} a chain of oscillators with coupling potential $V_{\rm KVM}$,  
by a simple transformation.

\subsection{The deterministic linear model}
\label{subsec:linmodel}
We consider here the specific case $V(r)=r^2/2$. The dynamics is then linear and can be solved analytically using Fourier transform. To simplify the exposition we consider the dynamics in infinite volume. We introduce the $k$th mode ${\widehat \eta} (k, \cdot)$ for $k \in {\bb T} = 
\mathbb{R}/\ZZ$, the one-dimensional torus of length $1$: 
\begin{equation*}
{\widehat \eta} (t,k) =\sum_{x \in \ZZ} \eta_x (t) \, e^{2 \ri \pi k x}. 
\end{equation*}
Then, the equations of motion are equivalent to the 
following decoupled system of first order differential equations:
\begin{equation*}
\frac{d{\widehat \eta}}{dt} (t,k)  = \ri \omega (k) \, {\widehat \eta} (t,k), 
\end{equation*}
where the dispersion relation $\omega (k)$ reads
\begin{equation*}
\omega (k) =- 2 \sin (2\pi k),
\end{equation*} 
and the group velocity $v_{\rm g}$ is
\begin{equation*}
v_{\rm g} (k)  = \omega' (k) = -4\pi \cos (2\pi k).
\end{equation*}
By inverting the Fourier transform, the solution can be written as
\[
\eta_{x} (t) = \int_{\TT} {\widehat \eta} (t,k) \, e^{- 2 \ri \pi k x} \, dk.
\]
Note also that the energy of the $k$th mode
\[
E_k (t)= \frac{1}{4\pi} |{\widehat \eta} (t,k)|^2 = E_k (0)
\]
is conserved by the time evolution, and that the total energy current 
$J^e= \sum_{x \in \ZZ} j^{e}_{x,x+1}$ takes the simple form
\begin{equation*}
J^e = \int_{\TT} v_{\rm g}(k) E_k \, dk.
\end{equation*}

We interpret the waves ${\widehat \eta} (k,t)$ as fictitious particles similar to phonons in solid state physics. In the absence of nonlinearities, they travel the chain without scattering. If the potential is non-quadratic, it may be expected that the nonlinearities produce a scattering responsible for the diffusion of the energy. Nevertheless, the conservation of the volume, which is expressed by
\begin{equation}
\label{eq:consj}
{\widehat \eta} (t, 0) = {\widehat \eta} (0,0),
\end{equation}
plays a crucial role. The identity (\ref{eq:consj}) is valid even if $V$ is not quadratic. It means that the $0$th mode is not scattered at all and crosses the chain ballistically. In fact, the modes with small wave number $k$ do not experience a strong scattering and they therefore contribute to the observed anomalous diffusion of energy.
For anharmonic chains of oscillators, a similar picture arises with ${\widehat \eta} (k)$ replaced by the phonons. As the conservation of momentum for these chains is responsible for the small scattering of phonons with small wave numbers, the conservation of the volume for the model considered in this paper is responsible for the small scattering of the waves ${\widehat \eta} (k)$ with small wave numbers.  

\subsection{Stochastic energy-volume conserving dynamics}
\label{subsec:sto}
We now consider energy-volume conserving stochastic perturbations of the deterministic dynamics generated by $\mc A$ or $\mathcal{A}_N$. 
The generator of the finite volume perturbed dynamics is written as
\begin{equation}
{\mc L}_N ={\mc A}_N  +\gamma {\mc S}_N,
\end{equation}
where ${\mc S}_N$ is the generator of the noise and $\gamma>0$ its intensity. 
The generator~${\mc S}_N$ reads
\begin{equation}
\label{eq:S}
({\mc S}_N f)(\eta) =\sum_{x=-N}^{N-1} \left[ f(\eta^{x,x+1}) -f(\eta) \right],
\end{equation} 
where $\eta^{x,x+1}$ is the configuration obtained from $\eta$ by exchanging the 
variables $\eta_x$ and $\eta_{x+1}$. 

By arguments similar to the one used to prove the well-posedness of the infinite deterministic
dynamics, it can be shown that the 
stochastic energy-volume conserving dynamics in infinite volume is also well defined (see \cite{FFL} for details). 
Its formal generator is given by  ${\mc L}={\mc A}+\gamma {\mc S}$ where
\begin{equation*}
({\mc S} f)(\eta) =\sum_{x \in \ZZ} \left[ f(\eta^{x,x+1}) -f(\eta) \right],
\end{equation*} 
and the corresponding Chapman-Kolmogorov equations (\ref{eq:CK}) are valid for this process
upon replacing $\mathcal{A}$ by $\mathcal{L}$. 
In particular the probability measures $\mu_{\beta,\lambda}$ are still invariant. 

%--------------------- Hyperbolic scaling ---------------------
\section{Hyperbolic scaling}
\label{sec:hyp_scaling}

We present in this section the hydrodynamic limit of the models described in the 
previous section. To this
end, we first need to define some thermodynamic quantities
useful to describe local equilibria (Section~\ref{sec:def_thermo}). We then informally describe
the expected hydrodynamic limit in Section~\ref{sec:informal_hydro_lim}, and conclude this 
section by stating precisely the convergence result (Theorem~\ref{th:hl}) in 
Section~\ref{sec:derivation_hydro}. 

\subsection{Definition of thermodynamic variables}
\label{sec:def_thermo}

Recall that the probability measures $\mu_{\beta, \lambda}$ form a family 
of invariant probability measures for the infinite dynamics defined in 
Section~\ref{sec:models}. 
The following thermodynamic relations (which are valid 
since we assumed that 
the partition function $Z$ is well defined on $(0,+\infty) \times {\mathbb R}$)
relate the chemical potentials $\beta, \lambda$ 
to the mean volume $v$ and the mean energy $e$ under $\mu_{\beta,\lambda}$:
\begin{equation}
\label{eq:tr}
\eqalign{
v(\beta,\lambda) =\mu_{\beta,\lambda} (\eta_x)= -\partial_{\lambda} \Big(\log Z(\beta,\lambda)\Big), \cr
e(\beta,\lambda)= \mu_{\beta,\lambda} (V(\eta_x))= -\partial_{\beta}\Big(\log Z(\beta,\lambda)\Big).
}
\end{equation} 
These relations can be inverted by a Legendre transform to express $\beta$ and~$\lambda$ 
as a function of~$e$ and~$v$. 
Define the thermodynamic entropy $S \, : \, (0,+\infty) \times \RR \to [-\infty,+\infty)$ as
\begin{equation*}
S(e,v)= \inf_{ \lambda \in \RR, \beta >0} \Big\{ \beta e + \lambda v + 
\log Z (\beta,\lambda) \Big\}.
\end{equation*}
Let ${\mc U}$ be the convex domain of $(0,+\infty) \times \RR$ where 
$S(e,v) >- \infty$ and $\mathring{\mc U}$ its interior. 
Then, for any $(e,v):=(e(\beta, \lambda),v(\beta,\lambda)) 
\in {\mathring{\mc U}}$, the parameters $\beta,\lambda$ can be obtained as 
\begin{equation}
\label{eq:14}
\beta = (\partial_e S ) (e,v), 
\qquad 
\lambda= (\partial_v S) (e,v).
\end{equation} 
We also introduce the tension 
$\tau(\beta,\lambda)= \mu_{\beta,\lambda} (V' (\eta_0))=-\lambda / \beta$.
Then, 
\begin{equation}
\label{eq:averages_of_currents}
\mu_{\beta,\lambda} (j_{x,x+1}^e)=-\tau^2, 
\qquad 
\mu_{\beta,\lambda} (j_{x,x+1}^v)=-2\tau.
\end{equation}
In the sequel, with a slight abuse of notation, 
we also write $\tau$ for $\tau(\beta(e,v), \lambda(e,v))$ 
where $\beta(e,v)$ and $\lambda(e,v)$ are defined by the relations~(\ref{eq:14}).

\subsection{Description of the hydrodynamic limit}
\label{sec:informal_hydro_lim}

Consider the finite \emph{closed} stochastic energy-volume dynamics with periodic boundary conditions, 
that is the dynamics generated by ${\mc L}_{N,{\rm{per}}} ={\mc A}_{N,{\rm{per}}} + \gamma {\mc S}_{N,{\rm{per}}}$ where
\begin{equation}
\label{eq:A_per}
\Big({\mathcal A}_{N,{\rm per}} f \Big)(\eta) = 
\sum_{x \in \TT_N} \left[V'(\eta_{x+1})-V'(\eta_{x-1}) \right] \partial_{\eta_x}f(\eta),
\end{equation}
and
\begin{equation*}
({\mc S}_{N,{\rm per}} f)(\eta) =\sum_{x \in \TT_N} \left[ f(\eta^{x,x+1}) -f(\eta) \right],
\end{equation*}
with $\TT_N = \RR/(N\ZZ)$ is the discrete torus of length $N$.  We choose to consider the dynamics on $\TT_N$ rather than on $\ZZ$ to avoid nontrivial technicalities.  We are interested in the macroscopic behavior of the two conserved quantities 
on a macroscopic time-scale $Nt$ as $N \to \infty$. 

We assume that the system is initially distributed according to a local Gibbs equilibrium  
state corresponding to a given energy-volume profile ${X_0} : \TT \to {\mathring{\mc U}}$: 
\begin{equation*}
X_0=\left(
\begin{array}{c}
{\mf e}_0\\
{\mf v}_0
\end{array} 
\right), 
\end{equation*}
in the sense that, for a given system size~$N$, 
the initial state of the system is described by the following product probability measure: 
\begin{equation}
\label{eq:Ges}
d\mu_{{\mf e}_0, {\mf v}_0}^N(\eta) = 
\prod_{x \in \TT_N} \frac{ \exp\left\{ -  {\mf \beta}_0 (x/N) V(\eta_x) 
-  {\mf \lambda}_0 (x/N) \eta_x \right\}}{Z ( {\mf \beta}_0 (x/N), {\mf \lambda}_0 (x/N))}\,
d\eta_x,
\end{equation} 
where $(\beta_0 (x/N), \lambda_0 (x/N))$ is actually a function of 
$({\mf e}_0 (x/N), {\mf v}_0 (x/N))$ through the relations~(\ref{eq:14}).

Starting from such a state, we expect the state of the system at time $t$ to be 
close, in a suitable sense, to a local Gibbs equilibrium measure 
corresponding to an energy-volume profile
\begin{equation*}
X(t,\cdot)=\left(
\begin{array}{c}
{\mf e}(t, \cdot)\\
{\mf v} (t,\cdot)
\end{array} 
\right),
\end{equation*}
satisfying a suitable partial differential equation with initial condition~$X_0$
at time $t = 0$. In view of~(\ref{eq:averages_of_currents}), and assuming local 
equilibrium, it is not difficult
to show that the expected partial differential equation 
is the following system of two conservation laws: 
%(note that the stochastic part 
%does not contribute, see the proof of Theorem~\ref{th:hl} 
%below for an explanation of this fact):
\begin{equation}
\label{eq:syslim}
\eqalign{
\partial_t {\mf e} -\partial_q \tau^2 =0,\cr 
\partial_t {\mf v} - 2 \partial_{q} \tau=0,
}
\end{equation} 
with initial conditions ${\mf e}(0,\cdot) = {\mf e}_0(\cdot), {\mf v}(0,\cdot) = {\mf v}_0(\cdot)$.
We write~(\ref{eq:syslim}) more compactly as  
\begin{equation*}
\partial_t X + \partial_q {\mf J}(X) =0, \qquad X (0,\cdot)=X_0 (\cdot),
\end{equation*}
with 
\begin{equation}
\label{eq:JJJ}
{\mf J} (X)= \left( 
\begin{array}{c}
-\tau^2 ( {\mf e}, {\mf v})\\
-2\tau ({\mf e}, {\mf v})
\end{array}
\right).
\end{equation}

The system of conservation laws (\ref{eq:syslim}) has other 
nontrivial conservation laws. In particular, the thermodynamic entropy~$S$ 
is conserved along a smooth solution of~(\ref{eq:syslim}):
\begin{equation}
\label{eq:consentropy}
\partial_t S ({\mf e},{\mf v})= 0.
\end{equation}
Since the thermodynamic entropy is a strictly concave function 
on $\mathring{\mc U}$, the system~(\ref{eq:syslim}) is strictly 
hyperbolic on $\mathring{\mc U}$ (see~\cite{S}). The two real 
eigenvalues of $(D{\mf J})({\bar \xi})$ are $0$ 
and $-\left[ \partial_{e} (\tau^2) + 2 \partial_v (\tau)\right]$, 
corresponding respectively to the two eigenvectors
\begin{equation}
\left(
\begin{array}{c}
-\partial_v \tau\\
\partial_e \tau
\end{array}
\right), \qquad
\left(
\begin{array}{c}
\tau\\
1
\end{array}
\right).
\end{equation}

It is well known that classical solutions to systems of $n \ge 1$ conservation laws develop shocks, even when starting from smooth initial conditions. Nevertheless, the Cauchy problem is locally well-posed in the Sobolev spaces $H^{s}(\TT)$ (for $s>3/2$). If we consider weak solutions rather than classical solutions, then a criterion is needed to select a unique, relevant solution among the weak ones. For scalar conservation laws ($n=1$), this criterion is furnished by the so-called entropy inequality and existence and uniqueness of solutions is fully understood. If $n\ge2$, only partial results exist (see~\cite{S,Br}). This motivates the fact that we restrict our analysis to smooth solutions before the appearance of shocks.

\subsection{Derivation of the hydrodynamic limit}
\label{sec:derivation_hydro}

We now turn to the question of deriving the system of conservation
laws as the hydrodynamic limit of the interacting particle system
under investigation. We assume that the potential $V$ satisfies the following

\begin{assumption}
\label{ass:V}
The potential $V$ is a smooth, non-negative function such that the partition function $Z(\beta, \lambda) = \int_{-\infty}^{\infty} \exp\left( -\beta V(r) -\lambda r \right)\, dr$ is well defined for $\beta>0$ and $\lambda \in \RR$ and  there exists a positive constant $C$ such that 
\begin{equation}
\label{eq:ass-pot1}
0 < V'' (r) \le C, 
\end{equation}
and
\begin{equation}
\label{eq:ass-pot2}
\limsup_{| r| \to + \infty} \frac{r V' (r)}{V(r)} \in(0, +\infty),
\end{equation}
\begin{equation}
  \label{eq:ass-pot3}
  \limsup_{| r| \to + \infty} \frac{[V' (r)]^2}{V(r)} < +\infty.  
\end{equation}
\end{assumption}

The hypothesis~(\ref{eq:ass-pot1}) allows to define easily the dynamics in infinite volume; (\ref{eq:ass-pot2}) is needed in the proof of Theorem~\ref{th:1}; (\ref{eq:ass-pot3}) ensures that the currents of the conserved quantities are bounded by the energy. This is useful to introduce a suitable cutoff for the derivation of hydrodynamic limits (see \cite[Section~3]{BO}).   

\emph{Provided} we can prove that the infinite volume
dynamics is ergodic in a suitable sense (see
Definition~\ref{def:ergo} below), then we can rigorously prove,
using the relative entropy method of Yau, that~(\ref{eq:syslim0}) is
indeed the hydrodynamic limit in the smooth regime, \textit{i.e.}
for times~$t$ up to the appearance of the first shock (see for
example~\cite{KL,TV}).
Recall indeed that a simple computation shows that the stochastic perturbation 
does not modify the hydrodynamic
limit since the effect of the latter is observed in the diffusive scale only.
  
In most cases, the derivation of hydrodynamic
limits is performed for stochastic interacting particle systems
which are trivially ergodic by construction. For deterministic
systems on the other hand, the ergodicity is extremely difficult to
prove. We are only able to show a weaker form of such ergodicity for
the process generated by ${\mc A}$ (with the additional assumption 
that the invariant measure is exchangeable). This weaker form is nonetheless sufficient to
show that the process generated by~${\mc L}$ is ergodic (see Theorem~\ref{th:1} below). 

As argued in~\cite{TV}, it turns out that when there are more than one conservation laws, the conservation of thermodynamic entropy (\ref{eq:consentropy}) must hold for the hydrodynamic limit to be well defined. This relation is indeed fundamental for Yau's method where, in the expansion of the time derivative of relative entropy, the cancelation of the linear terms is a consequence of the preservation of the thermodynamic entropy.

Averages with respect to the empirical energy-volume measure are defined, for continuous functions $G,H : \TT \to \RR$, as 
\[
\left(\begin{array}{c} 
{\mc E}_N (t,G) \cr 
{\mc V}_N (t,H) 
\end{array} \right) 
%= \int_\TT \left( \begin{array}{c} G \\ H \end{array}\right) d\pi^N(t,\cdot)
= \left( \begin{array}{c}
\displaystyle \frac1N \sum_{x \in \TT_N} G\left(\frac{x}{N}\right) \, V(\eta_x (t) ) \cr
\displaystyle \frac{1}{N} \sum_{x \in \TT_N} H\left(\frac{x}{N}\right) \, \eta_x (t)
\end{array} \right).
\]   
We can then state the following result.

\begin{theo}
\label{th:hl}
Fix some $\gamma > 0$ and consider 
the dynamics on the torus ${\bb T}_N$ generated by ${\mc L}_{N,{\rm per}}$ where the potential $V$ satisfies Assumption~\ref{ass:V}. Assume that the system is initially distributed according to a local Gibbs state~(\ref{eq:Ges}) with smooth energy profile ${\mf e}_0$ and volume profile ${\mf v}_0$. Consider a positive time $t$ such that the solution $({\mf e}, {\mf v})$ to (\ref{eq:syslim}) belongs to~${\mathring{\mc U}}$ and is smooth on the time interval~$[0,t]$. Then, for any continuous test functions $G,H:\TT \to \RR$, the following convergence in
probability holds as $N \to +\infty$:
\[
\Big({\mc E}_N (tN, G), {\mc V}_N (tN,H)\Big) \longrightarrow 
\left(\int_{\TT} G(q) {\mf e} (t,q) dq,  \int_{\TT} H(q) {\mf v} (t,q) dq\right).
\]
\end{theo}

The derivation of the hydrodynamic limits beyond the shocks for systems of conservation laws of dimension $n \geq 2$ is very difficult and is one of the most challenging problems in the field of hydrodynamic limits. The first difficulty is of course our poor understanding of the solutions to such systems. Recently, J. Fritz proposed in~\cite{F0} to derive hydrodynamic limits for hyperbolic systems (in the case $n=2$) by some extension of the compensated-compactness approach~\cite{Mu,Ta,DIP} to stochastic microscopic models. This program has been achieved in~\cite{FT} (see also \cite{F2}), where the authors derive the classical $n=2$ Leroux system of conservation laws. In fact, to be exact, only the convergence to the set of entropy solutions is proved, the question of uniqueness being left open. It remains nonetheless the best result available at this time. The proof is based on a strict control of entropy pairs at the microscopic level by the use of logarithmic Sobolev inequ
 ality estimates. It would be very interesting to extend these methods to systems such as the ones considered in this paper. 

%------------------- Ergodicity -------------
\section{Ergodicity}
\label{sec:erg}

We prove here the ergodicity of the stochastic dynamics, which is the
fundamental ingredient for the hydrodynamic limit.

\subsection{Definitions and notation}
In order to explain what is meant by ergodicity of the infinite volume dynamics 
we need to introduce some notation. 
For any topological space $X$ equipped with its Borel $\sigma$-algebra we 
denote by ${\mc P} (X)$ the convex set of probability measures on $X$.  
The relative entropy $H(\nu|\mu)$ of $\nu \in {\mc P} (X)$ with respect to 
$\mu \in {\mc P} (X)$ is defined as
\begin{equation}
\label{eq:ent009}
H(\nu | \mu) = \sup_{\phi} \left\{ \int \phi \, d\nu - 
\log \left( \int e^{\phi} \, d\mu \right) \right\},
\end{equation}
where the supremum is carried over all bounded measurable functions $\phi$ on $X$. 
Recall also the entropy inequality, which states that for every positive constant 
$a>0$ and every bounded measurable function $\phi$, it holds
\begin{equation}
\label{eq:enti}
\int \phi \, d\nu \le a^{-1} \, \left\{ \log \left( \int e^{a \phi} \, d\mu \right) 
+ H(\nu | \mu)  \right\}.
\end{equation}

Let $\theta_x, x \in \ZZ$, be the shift by $x$: $(\theta_x \eta)_z=\eta_{x+z}$. 
For any function $g$ on $\Omega$, $\theta_x g$ is the function such that $(\theta_xg)(\eta) 
= g(\theta_x \eta)$. For any probability measure $\mu \in {\mc P} (\Omega)$, 
$\theta_x \mu \in {\mc P} (\Omega)$ is the probability measure such that, 
for any bounded function $g: \Omega \to \RR$, 
it holds $\int_\Omega g \, d (\theta_x \mu)= \int_\Omega \theta_x g \, d\mu$. 
If $\theta_x \mu = \mu$ for any $x$ then $\mu$ is said to be translation invariant.

If $\Lambda $ is a finite subset of $\ZZ$ the marginal of $\mu \in {\mc P} (\Omega)$ 
on $\RR^{\Lambda}$ is denoted by $\mu |_{\Lambda}$. 
The relative entropy of $\nu \in {\mc P} (\Omega)$ with respect to $\mu \in {\mc P} (\Omega)$ 
in the box $\Lambda$ is defined by $H(\nu |_{\Lambda} \, | \, \mu |_{\Lambda} )$ 
and is denoted by $H_{\Lambda} (\nu| \mu)$. We say that a translation invariant probability 
measure $\nu \in {\mc P} (\Omega)$ has finite entropy density (with respect to $\mu$) 
if there exists a finite positive constant $C$ such that for any finite $\Lambda \subset \ZZ$, 
$H_{\Lambda} (\nu | \mu) \le C | \Lambda|$. In fact, if this condition is satisfied, then the limit 
\[
\overline{H} (\nu|\mu)=\lim_{|\Lambda| \to \infty} \frac{H_{\Lambda} (\nu | \mu) }{|\Lambda|}
\]
exists and is finite (see~\cite{FFL}). 
It is called the entropy density of $\nu$ with respect to $\mu$.  

We are now in position to define ergodicity.

\begin{definition}
\label{def:ergo}
We say that the infinite volume dynamics with infinitesimal generator ${\mc G}$ 
is \emph{ergodic} if the following claim is true:
If $\nu \in {\mc P} (\Omega)$ is a probability measure invariant by translation, 
invariant by the dynamics generated by ${\mc G}$ and with finite entropy density with respect to $\mu_{1,0}$, 
then $\nu$ is a mixture of the $\mu_{\beta,\lambda}, \beta>0, \lambda \in \RR$. 
\end{definition}

\subsection{Ergodicity of the stochastic dynamics}
We are not able to prove the ergodicity of the 
deterministic dynamics in general, but we can prove it under the additional 
assumption that the invariant measure is exchangeable.  

\begin{theo}
\label{th:1}
Assume that the potential $V$ satisfies~(\ref{eq:ass-pot1})-(\ref{eq:ass-pot2}). 
Let $\nu$ be a translation invariant measure with a finite local 
entropy density w.r.t. $\mu_{1,0}$ such that
\begin{equation}
\label{eq:s}
\forall f \in C^{1}_0 (\Omega), \qquad \int {\mc A} f \, d\nu =0.
\end{equation}
If $\nu$ is exchangeable then $\nu$ is a mixture of 
$\mu_{\beta,\lambda}$, $\beta>0$, $\lambda \in \RR$. 
\end{theo} 

The proof of Theorem~\ref{th:1} is provided in Section~\ref{sec:proof_Thm2}.
This result has an interesting consequence.

\begin{cor}
  The infinite volume dynamics generated by ${\mc L}$ is ergodic.
\end{cor}

The proof of this corollary is similar to the proof given in~\cite{FFL} (or~\cite{BO}). 
It is based on the fact that, if $\nu$ is invariant for ${\mc L}$, then it can be shown, 
by some entropy arguments, that $\nu$ is invariant separately for ${\mc A}$ and for ${\mc S}$. 
The invariance with respect to ${\mc S}$ implies that $\nu$ is 
exchangeable and we can then apply Theorem~\ref{th:1} to conclude. 

\subsection{Proof of Theorem~\ref{th:1}}
\label{sec:proof_Thm2}

We call ${\mc F}_{\rm inv}$ the $\sigma$-field generated 
by the $\theta_1$-invariant sets and ${\mc F}^2_{\rm inv}$ the $\sigma$-field 
generated by the $\theta_2$-invariant sets.
We denote by $\bar{\nu}$ the conditional measure $\nu(\cdot | {\mc F}^2_{\rm inv})$.

By the entropy inequality (\ref{eq:enti}), it is easy to show that 
\begin{equation*}
\nu(\eta_0) < +\infty, 
\qquad 
\nu (\eta_0 V'(\eta_0)) < +\infty.
\end{equation*}
Thus, the ergodic theorem gives the existence of the ${\mc F}_{\rm inv}^2$-measurable functions 
\begin{equation}
\label{eq:ergodic_limit}
\eqalign{
{\mc V}^1(\eta)=\lim_{\ell \to \infty} \frac{1}{2\ell +1} \sum_{|x| \le \ell} \eta_{2x+1}
= {\bar \nu}(\eta_1), \cr
\alpha^1 (\eta)= \lim_{\ell \to \infty}  \frac{1}{2\ell +1} \sum_{|x| \le \ell} 
(\eta_{2x+1} -{\mc V}^1) \, V' (\eta_{2x+1}) = 
{\bar \nu} \left[ (\eta_{1} -{\mc V}^1)\, V' (\eta_{1}) \right],
}
\end{equation}
where the convergence occurs in $L^1(\nu)$ and $\nu$-almost surely.
The random variable ${\mc V}^1$ can be considered as a constant under ${\bar \nu}$. 
Let us fix $x \in \ZZ$ and define
\[
f(\eta) = (\eta_{2x+1} - {\mc V}^1) \phi (\eta),
\]
where $\phi \in C_0^1 (\Omega)$ is a function depending 
only on the even sites $\{ \eta_{2z} \}_{z \in \ZZ}$. 
By Lemma~\ref{lem:22}, the conditional probability measure ${\bar \nu}$ is ${\mc A}$ invariant.
Therefore,
\begin{eqnarray*}
\int\mathcal{A}f \, d\bar{\nu} = 0
& = \int (V'(\eta_{2x+2}) -V'(\eta_{2x}) ) \phi \, d{\bar \nu} (\eta) \\
& \ \ + \sum_{z \in \ZZ} \int (V'(\eta_{2z+1}) -V'(\eta_{2z-1})) ({\eta_{2x+1}} -{\mc V}^1) \partial_{\eta_{2z}} \phi \, d{\bar \nu}\\
& = \int (V'(\eta_{2x+2}) -V'(\eta_{2x}) ) \phi \, d{\bar \nu} (\eta)\\
& \ \ +  \int (\eta_{2x+1} -{\mc V}^1) \left[V'(\eta_{2x+1})-V'(\eta_{2x-1}) \right] \partial_{\eta_{2x}} \phi \, d{\bar\nu}\\
& \ \ +   \int (\eta_{2x+1} -{\mc V}^1) \left[V'(\eta_{2x+3})-V'(\eta_{2x+1}) \right]\partial_{\eta_{2(x+1)}} \phi \, d{\bar\nu}\\
& \ \ +  \sum_{z \ne x,x+1} \int (\eta_{2x+1}- {\mc V}^1)  \left[ V'(\eta_{2z+1})-V' (\eta_{2z-1}) \right]\partial_{\eta_{2z}} \phi \, d{\bar \nu}.
\end{eqnarray*}
By Lemma~\ref{lem:ind}, ${\bar \nu}$ is exchangeable. 
By exchanging $\eta_{2z+1}$ and ${\eta}_{2z-1}$ in the terms appearing in the last sum above, 
we see that the sum over $z \neq x,x+1$ is actually equal to $0$.
Moreover, by Lemma \ref{lem:ind} again, the functions depending on even sites 
are independent of functions depending on odd sites under ${\bar \nu}$, so that
\begin{equation}
\label{eq:proof_Thm2}
\eqalign{
0 = \int (V'(\eta_{2x+2}) -V'(\eta_{2x}) ) \phi \, d{\bar \nu}\cr
\qquad + \left(\int (\eta_{2x+1} -{\mc V}^1)V'(\eta_{2x+1}) \, d{\bar \nu}\right)
\left(\int \left(\partial_{\eta_{2x}} -\partial_{\eta_{2(x+1)}}\right) \phi \, d{\bar\nu}\right)\cr
\qquad - \left(\int (\eta_{2x+1} -{\mc V}^1) V'(\eta_{2x-1}) d{\bar \nu} \right) 
\left( \int  \partial_{\eta_{2x}} \phi \, d{\bar\nu} \right) \cr
\qquad +  \left(\int (\eta_{2x+1} -{\mc V}^1)V'(\eta_{2x+3}) d{\bar \nu} \right)
\left( \int \partial_{\eta_{2(x+1)}} \phi \, d{\bar\nu} \right).
}
\end{equation}
By exchangeability of ${\bar \nu}$, it holds, for any $k \neq -1$:
\[
{\nu} \left( (\eta_{2x+1} -{\mc V}^1) V' (\eta_{2x -1}) | {\mc F}_{\rm inv}^2\right) 
= {\nu} \left( (\eta_{2x+2k +1} -{\mc V}^1) V' (\eta_{2x -1}) | {\mc F}_{\rm inv}^2\right).
\]
Thus,
\[
\fl {\bar \nu}  \left( (\eta_{2x+1} -{\mc V}^1) V' (\eta_{2x -1}) \right)
= \nu \left( \left. \frac{1}{2\ell} \sum_{|k| \leq \ell, \, k \neq -1} 
(\eta_{2x+2k +1} -{\mc V}^1) V' (\eta_{2x -1}) \right| {\mc F}_{\rm inv}^2\right).
\]
The $L^1(\nu)$ limit in~(\ref{eq:ergodic_limit}) then gives
\begin{equation*}
{\bar \nu}  \left( (\eta_{2x+1} -{\mc V}^1) V' (\eta_{2x -1}) \right)=0.
\end{equation*}
Similarly, it can be shown that 
${\bar \nu}  \left( (\eta_{2x+1} -{\mc V}^1) V' (\eta_{2x +3}) \right) = 0$.
Since
\[
\alpha^1 =  {\bar \nu} \left[ (\eta_{1} -{\mc V}^1)\, V' (\eta_{1}) \right]= {\bar \nu} \left[ (\eta_{2x+1} -{\mc V}^1)\, V' (\eta_{2x+1}) \right],
\]
(\ref{eq:proof_Thm2}) implies
\begin{equation}
\label{eq:final1}
\fl \left \{ \eqalign{
0 = \int (V'(\eta_{2(x+1)}) -V'(\eta_{2x}) ) \phi \, d{\bar \nu} (\eta) +
\alpha^1 \int \left(\partial_{\eta_{2x}} -\partial_{\eta_{2(x+1)}}\right) \phi \, d{\bar\nu},\\
0 = \int (\eta_{2x+1} -{\mc V}^1) \, d{\bar \nu}.
}\right.
\end{equation}

In the same way, it can be shown that
\begin{equation}
\label{eq:final2}
\fl \left\{ \eqalign{
0 = \int (V'(\eta_{2(x+1)+1}) -V'(\eta_{2x+1}) ) \phi \, d{\bar \nu} (\eta) 
+ \alpha^0 \int \left(\partial_{\eta_{2x+1}} -\partial_{\eta_{2(x+1)+1}}\right) \phi \, d{\bar\nu},\\
0 =\int (\eta_{2x} -{\mc V}^0) \, d{\bar \nu},
} \right.
\end{equation}
where $\phi = \phi ((\eta_{2z+1})_{z \in \mathbb{Z}}) \in C_0^1 (\Omega)$ is now a 
test function depending on the odd sites only, and ${\mc V}^0$, ${\alpha}^0$ 
are the ${\mc F}_{\rm inv}^2$ measurable functions defined by
\begin{equation*}
{\mc V}^0 =\nu (\eta_0 | {\mc F}_{\rm inv}^2), 
\qquad 
{\alpha}^0 = \nu \left( V' (\eta_0) (\eta_0 - {\mc V}^0) | {\mc F}_{\rm inv}^2 \right).
\end{equation*}
In fact, the $\theta_1$-invariance of $\nu$ gives
\begin{equation*}
{\mc V}^0 ={\mc V}^1 = \nu ( \eta_0 | {\mc F}_{\rm inv}):={\mc V}, 
\qquad 
{\alpha}^0 = {\alpha}^1= \nu( (\eta_0 -{\mc V}) V'(\eta_0) | {\mc F}_{\rm inv}):= \alpha.
\end{equation*}
By Lemma \ref{lem:alpha}, it holds $\alpha >0$ $\nu$-almost surely. 
Taking into account the fact that $(\eta_{2z})_{z \in \ZZ}$ and $(\eta_{2z+1})_{z \in \ZZ}$ are 
independent under ${\bar \nu}$, (\ref{eq:final1}) and (\ref{eq:final2}) 
allow to show, by Lemma \ref{lem:Gibbscar}, 
that the probability measure ${\bar \nu}$ is a product measure with marginals given by 
\begin{equation*}
  d {\bar \nu} (\eta_{x} \in [r,r+dr])= {\widetilde Z}({\mc V}, {\alpha})^{-1} \exp 
  \left( -\alpha V(r) -\lambda r \right)dr, 
\end{equation*} 
where ${\widetilde Z}({\mc V}, {\alpha})$ is a normalizing constant, 
and $\lambda:=\lambda (\mc V ,\alpha)$ is such that
\begin{equation*}
{\widetilde Z}({\mc V}, {\alpha})^{-1} \int r \, \exp \left( -\alpha V(r) -\lambda r \right)dr 
={\mc V}.
\end{equation*}

Let us summarize our results: Denoting by ${\mb P}$ the law of the random variables 
$(\alpha(\omega) , \mc V (\omega) ) \in (0, +\infty) \times \RR$ under $\nu$, 
we have proved that, for any bounded local function $f$,
\begin{equation*}
\nu (f) = \nu (\nu (f | {\mc F}_{\rm{inv}}) )= \nu \left[ \mu_{\alpha (\omega), \mc V(\omega)} (f) \right]= \int \mu_{\beta, \lambda} (f) \, d{\mb P} (\beta, {\lambda}).
\end{equation*}
This concludes the proof of Theorem~\ref{th:1}.

%\end{proof}

\begin{lemma}
\label{lem:ind}
Let $\nu$ be an exchangeable translation invariant measure with a finite entropy density 
w.r.t. $\mu_{1,0}$. Then, ${\bar \nu}(\cdot) = \nu (\cdot | {\mc F}_{\rm inv}^2)$ 
is exchangeable and under $\bar \nu$, the variables 
$(\eta_{2x})_{x \in \ZZ}$ and $(\eta_{2x+1})_{x \in \ZZ}$ are independent.
\end{lemma}

\begin{proof}
For a given $x \in \ZZ$, consider the function $T^{x,x+1}: \Omega \to \Omega$ defined as
$(T^{x,x+1} \eta)=\eta^{x,x+1}$, and denote by 
${\mc T}$ be the set of local transformation $T:\Omega \to \Omega$ 
which are obtained as compositions of transformations $T^{y,y+1}$, $y \in \ZZ$.
To show that $\bar \nu$ is exchangeable, we prove that, for a given $x \in \ZZ$
and for any bounded function~$g$, 
\[
\int g \circ T^{x,x+1} \, d{\bar \nu} = \int g \, d {\bar \nu}.
\]
This amounts to proving that, for any ${\mc F}_{\rm inv}^{2}$-measurable 
bounded function~$\rho$,
\[
\nu (g\rho) = \nu\Big( \nu\left(\left.g\right|{\mc F}_{\rm inv}^{2}\right)\rho\Big)
= \nu\Big( \nu\left(\left.g\circ T^{x,x+1}\right|{\mc F}_{\rm inv}^{2}\right)\rho\Big)
= \nu \left( g\circ T^{x,x+1}\rho \right).
\]
By exchangeability of $\nu$, it is therefore sufficient to show that any 
${\mc F}_{\rm inv}^{2}$-measurable bounded function $\rho$ 
is invariant by $T^{x,x+1}$. 

To prove this result, we write $\rho$ as an ergodic limit of local functions.
Observe first that $\rho = \lim_{k \to \infty} \nu( \rho_k | {\mc F}_{\rm inv}^2 )$, $\nu$ a.s. and in ${L}^1 (\nu)$, where the local functions $\rho_k$ are defined as 
$\rho_k = \nu (\rho \, | \, {\mc F}_{\Lambda_k})$, with ${\mc F}_{\Lambda_k}$ 
the $\sigma$-algebra generated by $\{ \eta_x \, ; \, x \in \Lambda_k \}$.
Besides, 
\[
\nu (\rho_k | {\mc F}_{\rm inv}^2) = \lim_{\ell \to \infty} 
\frac{1}{2\ell +1} \sum_{j=-\ell}^{\ell} \theta_{2j} \rho_k,
\]
so that
\begin{equation}
  \label{eq:approximation_rho}
  \rho = \lim_{k \to \infty} \lim_{\ell \to \infty} \frac{1}{2\ell +1} 
  \sum_{j=-\ell}^{\ell} \theta_{2j} \rho_k.
\end{equation}
Now, since $\rho_k$ is a local function, $(\theta_{2j} \rho_k) \circ T^{x,x+1} 
= \theta_{2j} \rho_k$ for $|j|$ large enough. Therefore, 
\[
\left[ \frac{1}{2\ell +1} \sum_{j=-\ell}^{\ell} \theta_{2j} \rho_k \right] \circ T^{x,x+1} 
= \frac{1}{2\ell +1} \sum_{j=-\ell}^{\ell} \theta_{2j} \rho_k +{\mc O} (\ell^{-1}),
\]
which, together with~(\ref{eq:approximation_rho}), gives indeed $\rho \circ T^{x,x+1} = \rho$.

We turn now to the second part of the lemma. Let $\phi^0$ (resp. $\phi^1$) be 
a bounded local measurable function depending only on the 
even (resp. odd) sites. 
We have to show that
\begin{equation}
\label{eq:inde}
\nu \left( \phi^0 \phi^1 | {\mc F}_{\rm inv}^2\right)
=\nu \left( \phi^0 | {\mc F}_{\rm inv}^2\right)\nu \left( \phi^1 | {\mc F}_{\rm inv}^2\right).
\end{equation}
Fix $k, j \in \ZZ$ and consider a local transformation $T \in {\mc T}$ 
such that $\phi^0 \circ T = \theta_{2k} \phi^0$ and $\phi^1 \circ T = \theta_{2j} \phi^1$. 
This is possible because $\phi^0$ (resp. $\phi^1$) depends only on the even (resp. odd) sites.
Now, for any ${\mc F}_{\rm inv}^{2}$-measurable positive bounded function $\rho$,
\[
\fl \nu \Big( \nu \left( \phi^0 \phi^1 | {\mc F}_{\rm inv}^2\right) \rho \Big)
= \nu \left( \phi^0 \phi^1 \rho \right)
= \nu \left(( \phi^0 \circ T) \; (\phi^1 \circ T) \; (\rho \circ T) \right)
= \nu \left( \theta_{2k} \phi^0 \; \theta_{2j} \phi^1 \; \rho  \right),
\]
where the second equality is a consequence of the exchangeability of $\nu$, 
and the third one is obtained thanks to the invariance of $\rho$ by $T$ proved above. 
It follows that
\begin{equation*}
\nu \Big( \nu \left( \phi^0 \phi^1 | {\mc F}_{\rm inv}^2\right) \rho \Big)
= \nu \left( \left( \frac{1}{2\ell +1} \sum_{k=-\ell}^\ell \theta_{2k} \phi^0\right) 
\left( \frac{1}{2m +1} \sum_{j=-m}^m \theta_{2j} \phi^1\right) \rho \right).
\end{equation*}
Taking first the limit $\ell \to \infty$ and then the limit $m \to \infty$, 
we obtain, by the ergodic theorem,
\[
\nu \Big( \nu \left( \phi^0 \phi^1 | {\mc F}_{\rm inv}^2\right) \rho \Big) =
\nu \Big( \nu \left( \phi^0 | {\mc F}_{\rm inv}^2\right)
\nu \left( \phi^1 | {\mc F}_{\rm inv}^2\right) \rho \Big),
\]
which is indeed~(\ref{eq:inde}).
\end{proof}

\begin{lemma}
\label{lem:22}
Let $\nu$ a probability measure invariant by translation, invariant 
by ${\mc A}$ and with finite entropy density. Let $f \in C_0^1 (\Omega)$ 
and $\rho$ be a ${\mc F}_{\rm inv}^2$ bounded measurable function. Then
\begin{equation}
  \label{eq:A_inv_cond_meas}
  \int ({\mc A} f)\,  \rho \, d\nu = \nu\Big( \nu({\mc A} f | {\mc F}_{\rm inv}^2) \rho \Big) = 0.
\end{equation}
In other words, $\nu(\cdot| {\mc F}_{\rm inv}^2)$ is ${\mc A}$-invariant.
\end{lemma}

\begin{proof}
The tail $\sigma$-field ${\mc F}_{\rm{tail}}$ is defined by ${\mc F}_{\rm{tail}}= \cap_{k \ge 1} {\mc F}_{\Lambda_k^c}$ where ${\mc F}_{\Lambda_k^c}$ is the $\sigma$-field generated by $\{\eta_x\,;\, |x|>k\}$. Any ${\mc F}_{\rm{inv}}^2$-measurable function $\rho$ coincides $\nu$ almost surely with a ${\mc F}_{\rm{tail}}$-measurable function (this can be seen for example as a consequence of (\ref{eq:approximation_rho})). In particular, for any $k \in \ZZ$, $\partial_{\eta_k} \rho =0$ $\nu$ a.s. . Thus, if $f \in C_0^1 (\Omega)$, then $\rho f \in C_0^1 (\Omega)$ and by (\ref{eq:s}) applied to the latter, we obtain $ 0= \int {\mc A} (\rho f) \, d\nu= \int ({\mc A} f)\,  \rho \, d\nu$. 
\end{proof}

\begin{lemma}
\label{lem:alpha}
Let $\alpha= \nu \left( (\eta_0 -{\mc V})V' (\eta_0) | {\mc F}_{\rm inv}\right)$ 
where ${\mc V} =\nu (\eta_0 | {\mc F}_{\rm inv})$. Then, 
\begin{equation}
\nu (\{ \alpha \le 0 \}) = 0.
\end{equation}
\end{lemma}

\begin{proof}
Let ${\tilde \nu}=\nu (\cdot |{\mc F}_{\rm{inv}})$ and remark first that $\alpha = {\tilde \nu} \left((\eta_0 -\eta_1) V' (\eta_0)\right)$,
so that the exchangeability of $\nu$ implies $\alpha = {\tilde \nu} \left((\eta_1 -\eta_0) V' (\eta_1) \right)$,
and
\[
\alpha=\frac{1}{2} \, {\tilde \nu} \left((\eta_0 -\eta_1) (V' (\eta_0) -V' (\eta_1))\right).
\]
The convexity of $V$ already gives $\alpha \ge 0$.  It remains to show that $\nu (\{ \alpha = 0 \}) = 0$. To prove this, it is sufficient to show that the restriction of ${\tilde \nu}|_{\Lambda}$ to the box $\Lambda=\{0,1\}$, has ($\nu$ a.s.) a density with respect to the Lebesgue measure. In fact, we claim that the relative entropy of ${\tilde \nu}|_{\Lambda}$ w.r.t. $\mu_{0,1}|_{\Lambda}$ is finite, which implies the existence of the desired density. Indeed, consider any non-negative ${\mc F}_{\rm{inv}}$-measurable function $\rho$ and any positive bounded ${\mc F}_{\Lambda}$-measurable function $\varphi$. By the same argument as in the proof of Lemma \ref{lem:22}, $\rho$ is ${\mc F}_{\rm{tail}}$-measurable. Consequently, by using this and the definition of the conditional expectation, we have
\begin{equation*}
\nu (\rho\,  {\tilde \nu}|_{\Lambda} (\varphi))= \nu(\rho\,  {\tilde \nu} (\varphi))=\nu (\rho \varphi) = \nu( \rho \, \nu|_{\Lambda} (\varphi)).
\end{equation*}
By the variational definition (\ref{eq:ent009}) of the relative entropy and the previous equality, we get
\begin{equation*}
\nu (\rho\,  {\tilde \nu}|_{\Lambda} (\varphi)) \le \nu \left( \rho \left[ H \left( \nu|_{\Lambda} \, \Big| \, \mu_{0,1} |_{\Lambda} \right) + \log \left(\int e^{\varphi} d\mu_{0,1} \right) \right]\right).
\end{equation*}    
By assumption, $\nu$ has a finite entropy density so that  $H\left( \nu|_{\Lambda} \, \Big| \, \mu_{0,1} |_{\Lambda} \right) <+\infty$. Since the previous inequality is valid for any ${\mc F}_{\rm{inv}}$-measurable function $\rho \ge 0$, the claim $H\left( {\tilde \nu}|_{\Lambda} \, \Big| \, \mu_{0,1} |_{\Lambda} \right) <+\infty$ follows by the variational formula of the relative entropy (\ref{eq:ent009}).

\end{proof}

\begin{lemma}
\label{lem:Gibbscar}
Let  $V:\RR \to \RR$ be a potential satisfying Assumption \ref{ass:V} and $\beta>0$ a constant.  Let $\mu$ be a probability measure on $\Omega$ satisfying, for any $j \in \ZZ$ and any $\varphi \in C_0^1 (\Omega)$,
\begin{equation*}
\int (V' (\eta_{j+1}) -V' (\eta_{j}) ) \,\varphi \, d{\mu} + {\beta}^{-1} \int \left( \partial_{\eta_{j}} \varphi  -\partial_{\eta_{j+1}} \varphi \right) \, d{\mu} \, = \,0.
\end{equation*} 
Then, $\mu$ is a product probability measure whose marginals are given by $\left[Z (\beta, \lambda)\right]^{-1}  \exp (-\beta V (\eta_j) - \lambda \eta_j) \, d\eta_j$ where $\lambda$ is such that $v(\beta,\lambda)$ defined in (\ref{eq:tr}) coincides with $\mu (\eta_j)$.
\end{lemma}

\begin{proof}
Define $ \psi_\ell (\eta) = e^{\sum_{i=1}^{\ell} (\beta V(\eta_i) + \lambda \eta_i)}$, where the value of $\lambda$ is chosen so that 
\[
v(\beta,\lambda)=\frac{\displaystyle \int_{\RR} \eta_0 \, e^{-\beta V(\eta_0) - \lambda \eta_0} \, d\eta_0}
     {\displaystyle \int_{\RR} e^{-\beta V(\eta_0) - \lambda \eta_0} \, d\eta_0} = \mu (\eta_0). 
\]
Choosing $\phi(\eta) = \chi(\eta) \psi_\ell (\eta)$ with $\chi(\eta)$
a local compactly supported smooth function, we obtain, for any $j=1, \dots, \ell-1$,
\begin{equation*}
\beta^{-1} \int \left(\partial_{\eta_{j}} \chi
     - \partial_{\eta_{j+1}} \chi \right)  \psi_\ell (\eta) \; d\mu (\eta) = 0 \ .
\end{equation*}
We now consider 
$$
\chi(\eta) = \chi_b(\eta) g\left(\sum_{i=1}^{\ell} \eta_i\right) \chi_0 (\eta_{1}, \dots, \eta_{\ell})
$$
 where $\chi_b$ is a local function not depending on $\eta_{1}, \dots, \eta_{\ell}$, and $g$ is a smooth function on $\RR$. Since 
$$ \partial_{\eta_{j}} \chi(\eta) - \partial_{\eta_{j+1}} \chi(\eta) = \chi_b(\eta) g\left(\sum_{i=1}^{\ell} \eta_i\right) \left(\partial_{\eta_{j}} \chi_0 (\eta)  - \partial_{\eta_{j+1}} \chi_0 (\eta)\right)$$ if $j=1,\dots, \ell-1$, we can further condition on $\sum_{k=1}^{\ell} \eta_k = \ell u$ and on the
exterior configuration $\{\eta_i, i\neq {1},\dots,{\ell}\}$, and obtain, for all $j=1, \dots, \ell-1$, 
\begin{equation}
\label{eq:transleb}
\fl \int\left[\partial_{\eta_{j}} \chi_0 (\eta) - \partial_{\eta_{j+1}} \chi_0 (\eta)\right] \psi_\ell (\eta) \; \mu\left(d\eta_{1}, \dots, d\eta_{\ell}\left| \sum_{k=1}^{\ell} \eta_k = \ell u, \eta_i, \ i\neq {1}, \dots, {\ell}  \right.\right) = 0.
\end{equation}
These relations allow to show that the Borel measure
\begin{equation*}
\psi_\ell (\eta)\mu\left(d\eta_{1}, \dots, d\eta_{\ell}\left|\sum_{k=1}^{\ell} \eta_k = \ell u, \eta_i,\ i\neq {1},\dots, {\ell}\right.\right), %= e^{-\ell \lambda u}\prod_{i=1}^\ell e^{\beta V(\eta_i)} \mu\left(d\eta_{1}, \dots, d\eta_{\ell}\Big| \sum_{k=1}^{\ell} \eta_k = \ell u, \eta_i,\ i\neq {1},\dots,{\ell}\right)
\end{equation*}
which has support on the hyperplane ${\mc H}_\ell=\{ (\eta_{1}, \dots, \eta_{\ell}) : \sum_{k=1}^{\ell} \eta_k = \ell u\}$, is invariant by any translation of $\mathcal{H}_\ell$.
By Theorem~2.20 in~\cite{Ru}, this measure is therefore the Lebesgue measure on ${\mc H}_\ell$ 
up to a multiplicative constant. This multiplicative constant is fixed by ensuring that
the correct value of $\mu(\eta_0)$ is recovered, so that finally
\begin{equation*}
  \mu(d\eta_{1}, \dots, d\eta_{\ell}| \eta_i,\ i\neq {1},\dots,{\ell}) =  \frac{e^{-\sum_{i=1}^{\ell} 
    (\beta V(\eta_i) + \lambda \eta_i)}}{Z(\beta, \lambda)^\ell} 
  d\eta_{1} \dots d\eta_{\ell}.
\end{equation*}
%{\comm{Il manque ici un argument pour moi...}}
This concludes the proof of the lemma.
\end{proof}

\subsection{Derivation of a triplet of compressible Euler equations from a perturbed Hamilton dynamics}
\label{sec:secgeneral}

It is also natural to consider the Hamiltonian dynamics (\ref{eq:generaldynamics}) perturbed by the random exchange of neighboring velocities $p_{x}$ and $p_{x+1}$, so that $\sum_x p_x$, $\sum r_x$ and $\sum_x (U(p_x) +V(r_x))$ are preserved during the time evolution. This has been considered in \cite{FFL} in the case of a quadratic kinetic energy $U(p)=p^2/2$. In the case of a non-quadratic kinetic energy $U$, it is possible to adapt the arguments of \cite{FFL} to show that the corresponding dynamics is ergodic and that the triplet of compressible Euler equations (\ref{eq:syslim0}) is obtained as a hydrodynamic limit by Yau's method. We give only a sketch of the proof of the ergodicity of the perturbed dynamics.

Let $U$ and $V$ be two convex potentials satisfying Assumptions \ref{ass:V}. We state the dependence of the partition function with respect to the potential $V$ (resp. $U$) by the notation $Z_{V}$ (resp. $Z_{U})$. 

We consider the set ${\widetilde \Omega}=\Omega \times \Omega$ equipped with its natural product topology and its Borel $\sigma$-field. A typical configuration $\omega \in {\widetilde \Omega}$ is denoted by $\omega=(\bm{r} , \bm{p})$.  Let $\{\omega (t)\}_{t \ge 0}= \{ \bm{r} (t), \bm{p} (t) \}_{t \ge 0}$ be the the dynamics generated by ${\mc L} = {\mc A} + \gamma {\mc S}$, $\gamma >0$, with
\begin{equation*}
{\mc A} = \sum_{x \in \ZZ} (V' (r_x) -V' (r_{x-1}) \partial_{p_x} + \sum_{x \in \ZZ} (U' (p_{x+1}) -U' (p_x) ) \partial_{r_x}
\end{equation*}
and ${\mc S}$ the momenta exchange noise operator, acting on test 
functions~$f \, : \, \Omega \times \Omega \to \RR$ as
\begin{equation*}
({\mc S} f)(\omega)=\sum_{x \in \ZZ} \left [ f(\bm{r}, {\bm{p}}^{x,x+1}) - f(\bm{r}, \bm{p}) \right],
\end{equation*}
where ${\bm{p}}^{x,x+1}$ is the configuration obtained from $\bm{p}$ by exchanging the momentum $p_x$ with $p_{x+1}$. Observe that ${\mc L}, {\mc A}$ and ${\mc S}$ conserve the energy, the momentum and the deformation. All the equilibrium Gibbs measures $\{ \nu_{\beta, \lambda, \lambda'} \, , \beta>0, \lambda, \lambda' \in \RR\}$ are invariant. 

\begin{theo}
\label{th:thgeneral}
Assume that the potentials $U,V$ satisfy Assumptions \ref{ass:V}. The infinite volume dynamics generated by ${\mc L}$ is ergodic.
\end{theo}

\begin{proof}
Let $\nu$ be a translation invariant probability measure with finite entropy density invariant for ${\mc L}$. We have to show this is a mixture of the $\nu_{\beta, \lambda,\lambda'}$. As stated before, it can be shown by entropy arguments that $\nu$ is separately invariant for ${\mc A}$ and for ${\mc S}$. The invariance w.r.t. ${\mc S}$ implies that the law of $\bm p$ is exchangeable under $\nu$. Thus, we are reduced to proving that if $\nu$ is invariant for ${\mc A}$, shift invariant with finite entropy density and such that the law of $\bm{p}$ is exchangeable, then $\nu$ is a mixture of the $\nu_{\beta,\lambda,\lambda'}$'s. 

Let $\nu$ as above and call ${\mc F}_{\rm{inv}}$ the $\sigma$-field of invariant events for the shift. As in Lemma~\ref{lem:ind} or in the proof of Theorem~2.1 of \cite{FFL}, one can show that $\bm p$ and $\bm r$ are independent under ${\bar \nu}(\cdot)  = \nu(\cdot | {\mc F}_{\rm{inv}})$. 

Moreover, since an invariant set is, up to a $\nu$-negligible set, a tail invariant set, as argued in Lemma \ref{lem:22}, it can be proved that ${\bar \nu}$ is invariant for ${\mc A}$. Observe also that the law of ${\bm p}$ under ${\bar \nu}$ is exchangeable. Now, we can identify the distribution of ${\bm r}$ as follows. The thermodynamic relations give a one to one correspondence between the averages of the momentum and kinetic energy, $(\pi, T)$, and the chemical potentials, $\beta>0, \lambda \in \RR$, through the relations
\begin{equation*}
\pi (\beta,\lambda) = -\partial_{\lambda} \log Z_U (\beta,\lambda), \quad T(\beta, \lambda) = - \partial_{\beta} \log Z_U (\beta, \lambda).
\end{equation*}
Let $\beta, \lambda$ be the ${\mc F}_{{\rm{inv}}}$-measurable functions defined by the relations 
\begin{equation}
\label{eq:tr111}
\pi (\beta, \lambda)= {\bar \nu} (p_k), \quad T (\beta, \lambda) = {\bar \nu} (U(p_k)).
\end{equation}

Consider a test function of the form $\psi({\bm r},{\bm p})= \varphi ({\bm r}) (p_j -\pi)$ for a smooth function $\varphi$. Since $\int {\mc A} \psi \, d{\bar \nu} =0$, by using the independence of $\bm r$ and $\bm p$ under ${\bar \nu}$ and the exchangeability of the law of $\bm p$ under $\bar \nu$, we obtain easily that
\begin{equation}
\label{eq:eq789}
\int (V' (r_j) -V' (r_{j-1}) ) \varphi \, d{\bar \nu} + {\hat \beta}^{-1} \int \left( \partial_{r_{j-1}} \varphi  -\partial_{r_j} \varphi \right) d{\bar \nu} =0
\end{equation}
where ${\hat \beta}^{-1}= \int U' (p_j) (p_j -{\pi}) d{\bar \nu}$. The fact that ${\hat \beta}^{-1}$ is well defined, i.e. ${\hat \beta}$ is stricly positive $\nu$ a.s., can be proved as in Lemma \ref{lem:alpha}.

By Lemma \ref{lem:Gibbscar}, the equation (\ref{eq:eq789}) is sufficient to identify the distribution of ${\bm r}$ under ${\bar \nu}$ as a product measure $\prod_{x \in \ZZ} Z_V ({\hat \beta}, {\hat \lambda}') \, \exp\{ -{\hat \beta} V(r_x) -{\hat \lambda}' r_x  \} \, dr_x $. Moreover, ${\hat \beta}$ and ${\hat \lambda}'$ satisfy
\begin{equation*}
-\partial_{\lambda} \log Z_V ({\hat \beta}, {\hat \lambda}') ={\bar \nu} (r_k), \quad -\partial_{\beta} \log Z_V ({\hat \beta}, {\hat \lambda}') ={\bar \nu} ( V(r_k)).
\end{equation*}  

Setting $u={\bar \nu} (r_k)$ and using a test function of the form $\varphi ({\bm p}) (r_j -u)$, a similar computation shows that the law of ${\bm p}$ under ${\bar \nu}$ is in the form $\prod_{x \in \ZZ} Z_U ({\widetilde \beta}, {\widetilde \lambda}) \, \exp\{ -{\widetilde \beta} U(p_x) -{\widetilde \lambda} p_x  \} \, dp_x$ where ${\widetilde \beta}^{-1} = \int V' (r_j) (r_j -u) d{\bar \nu}$ and
\begin{equation*}
-\partial_{\lambda} \log Z_U ({\widetilde \beta}, {\widetilde \lambda}) ={\bar \nu} (p_k), \quad -\partial_{\beta} \log Z_U ({\widetilde \beta}, {\widetilde \lambda}) ={\bar \nu} ( U(p_k)).
\end{equation*}  
As above, ${\widetilde \beta}$ is $\nu$ a.s. strictly positive. In view of (\ref{eq:tr111}), we have ${\widetilde \beta}=\beta$ and ${\widetilde \lambda}=\lambda$. Injecting these informations in the definition of ${\hat \beta}^{-1}$, we obtain ${\hat \beta}=\beta$. 

We have therefore proved that 
$\bar \nu = \nu_{\beta, \lambda, {\hat \lambda}'}$ and the conclusion easily follows.   
\end{proof}

%---------------super-diffusive -------
\section{(Sub)Diffusive scaling}
\label{sec:diffusive}

The hydrodynamic limit is nothing but a law of large numbers. The second step of the study consists in looking at the fluctuations. As explained in the introduction, we expect that fluctuations appear at a shorter time-scale than the diffusive one. The study of the fluctuation field~(\ref{eq:696}) for the deterministic system is very difficult. A more tractable quantity which allows to decide whether $\alpha=1$ or $\alpha<1$ is the diffusivity ${\mc D}_{\beta,\lambda,\lambda'} (t)$. 

To define this quantity, we introduce some notation. The equilibrium compressibility matrix~${\widetilde \chi}$ is the symmetric matrix defined by:
\begin{equation}
\label{eq:chi0}
{\widetilde \chi}= \left( 
\begin{array}{ccc}
\mu_{\beta,\lambda,\lambda'} ({\mc E}_0 ; {\mc E}_0)& \mu_{\beta,\lambda,\lambda'} ({\mc E}_0 ; {p}_0)& \mu_{\beta,\lambda,\lambda'} ({\mc E}_0 ; {r}_0)\\
\mu_{\beta,\lambda,\lambda'} ({p}_0 ; {\mc E}_0)& \mu_{\beta,\lambda,\lambda'} ({p}_0 ; {p}_0)& \mu_{\beta,\lambda,\lambda'} ({p}_0 ; {r}_0)\\
\mu_{\beta,\lambda,\lambda'} ({r}_0 ; {\mc E}_0)& \mu_{\beta,\lambda,\lambda'} ({r}_0 ; {p}_0)& \mu_{\beta,\lambda,\lambda'} ({r}_0 ; {r}_0)\\
\end{array}
\right).
\end{equation}
Here, $\mu_{\beta,\lambda,\lambda'} (f;g)$ denotes the covariance of the two functions $f$ and $g$ w.r.t. $\mu_{\beta, \lambda,\lambda'}$. Let $\widehat{I}_{x.x+1}$ be the normalized current associated to the three conservation laws. We do not give a precise definition of the latter for the deterministic dynamics but the reader will translate easily in this case the one we give in~\eref{eq:normalized_current} below for the energy-volume conserving dynamics. Roughly speaking, ${\widehat I}_{x,x+1}$ is obtained from the usual microscopic currents of the three conserved quantities by a change of frame prescribed by the linearized flow (\ref{eq:lf00})  of the hydrodynamic equations. Let ${\widetilde C}_{\beta,\lambda,\lambda'}$ be the current-current correlation function defined as
\[
{\widetilde C}_{\beta,\lambda,\lambda'} (t)= \sum_{x \in \ZZ} \left \langle \,  \widehat{ {I}}_{x,x+1} (t)   \, \widehat{ {I}}_{0,1} (0)^* \, \right \rangle_{\beta,\lambda}.
\]
Then, the diffusivity is given by
\begin{equation}
\label{eq:diff00}
{\widetilde {\mc D}}_{\beta,\lambda,\lambda'} (t)= {\widetilde \chi}^{-1} \int_0^{\infty} \left( 1-\frac{s}{t} \right)^+ {\widetilde C}_{\beta,\lambda,\lambda'} (s) ds.
\end{equation}  

When the current-current correlation function is integrable, the limit 
when $t\to +\infty$ of the diffusivity is well defined, and the following Green-Kubo formula
is obtained:
\begin{equation}
\label{eq:diff222}
{\widetilde {\mc D}}_{\beta,\lambda,\lambda'}^\infty = 
\lim_{t \to \infty} {\widetilde{\mc D}}_{\beta,\lambda,\lambda'} (t) =  {\widetilde{\chi}}^{-1} \int_0^{\infty} {\widetilde C}_{\beta,\lambda,\lambda'} (s) \, ds.
\end{equation}
The existence of the above limit depends on the time decay of ${\widetilde C}_{\beta,\lambda,\lambda'}$. For a diffusive (resp. super-diffusive) behavior, ${\widetilde{\mc D}}_{\beta,\lambda}  (t)$ is of order ${\mc O} (1)$ (resp. of order ${\mc O} (t^{1-\alpha})$). 

The super-diffusive behavior of the deterministic dynamics ($\gamma = 0$) can be proved easily for the linear dynamics, \emph{i.e.} $V(r)=r^2/2$, and the Kac-van-Moerbecke dynamics, \emph{i.e.} $V(r)=e^{-r}+r-1$. In these cases, the dynamics is actually ballistic since the diffusivity is of order ${\mc O} (t)$. For the linear dynamics, this follows from the fact that the total current is a constant of motion. For the Kac-van-Moerbecke dynamics, an application of the Mazur inequality (see~\cite{M}), as done in~\cite{Z} and \cite{BO0} for the Toda lattice, shows that the diffusivity is also of order ${\mc O} (t)$. Apart from these two cases, showing a super-diffusive behavior for general potentials remains challenging.

We now turn to the energy-volume conserving dynamics. In this case, we also do not know, in general, what the large time behavior of the diffusivity is. Nevertheless, when $V(r)=r^2/2$, we can compute explicitly the value of the current-current correlation function and deduce the behavior of the diffusivity. However, before stating precisely the result, we first need to modify the definition~(\ref{eq:diff00}) of the diffusivity to account for the stochastic perturbation. We have now only two conserved quantities. The local energy-volume conservation is expressed by the formulas
\begin{equation*}
\eqalign{
V (\eta_x (t)) -V(\eta_x (0)) = - \nabla \left[ \int_{0}^t j_{x-1,x}^{e, \gamma} (s) ds + M_{x-1}^{e,\gamma} (t) \right], \cr
\eta_x (t) - \eta_x (0) = - \nabla \left[ \int_{0}^t j_{x-1,x}^{v, \gamma} (s) ds + M_{x-1}^{v,\gamma} (t) \right],
}
\end{equation*}
where the instantaneous currents $j_{x,x+1}^{e, \gamma}$, $j_{x,x+1}^{v,\gamma}$ are
\begin{equation}
\label{eq:current_gamma}
j_{x,x+1}^{e,\gamma} = j_{x,x+1}^{e} -\gamma \nabla \left[ V(\eta_x) \right],
\qquad
j_{x,x+1}^{v,\gamma}=  j_{x,x+1}^{v} -\gamma \nabla \left[ \eta_x \right],
\end{equation}
while the local martingales $M_{x}^{e,\gamma} (t) , M_{x}^{v,\gamma} (t)$ read
\[
\eqalign{
M_{x}^{e,\gamma}(t) = \int_0^t \left[V(\eta_{x+1}(s^-))-V(\eta_{x}(s^-))\right] \, d\left[
N_{x-1,x}(s)-\gamma s\right], \\
M_{x}^{v,\gamma}(t) = \int_0^t \left[\eta_{x+1}(s^-)-\eta_{x}(s^-)\right] \, d\left[
N_{x-1,x}(s)-\gamma s\right],
}
\]
where $(N_{x-1,x}(t))_{x \in \mathbb{Z}}$ are independent Poisson processes of intensity~$\gamma$.
Denoting by $[X,Y]_t$ the quadratic variation at time $t$ between two adapted processes $X$ and 
$Y$, it is easy to see that the martingales $M_x^{e,\gamma}$, $M_x^{v,\gamma}$ are such that
\begin{equation}
\label{eq:martqv}
\eqalign{
\left[ M_y^{a,\gamma} , M_z^{b,\gamma}  \right]_t =0, \qquad a,b \in \{e,v\}, \; y \ne z, \cr
\left[ M_x^{e,\gamma}, M_x^{e,\gamma} \right]_t = \gamma \int_0^t \left[ V(\eta_{x+1} (s)) - V(\eta_x (s)) \right]^2 ds,\cr
\left[ M_x^{v,\gamma}, M_x^{v,\gamma} \right]_t = \gamma \int_0^t \left[ \eta_{x+1} (s) -\eta_x (s)  \right]^2 ds,\cr
\left[ M_x^{e,\gamma}, M_x^{v,\gamma} \right]_t = \gamma \int_0^t \left[ V(\eta_{x+1} (s)) -V(\eta_x (s)) \right] \left[ \eta_{x+1} (s) -\eta_x (s)  \right] ds.
}
\end{equation}

Let us also introduce the equilibrium compressibility matrix~$\chi$:
\begin{equation}
\label{eq:chi}
\chi=
 \left( 
\begin{array}{cc}
\mu_{\beta,\lambda} (V(\eta_0); V(\eta_0)) & \mu_{\beta,\lambda} (V(\eta_0) ; \eta_0)\\
\mu_{\beta,\lambda} (\eta_0; V(\eta_0)) & \mu_{\beta,\lambda} (\eta_0 ; \eta_0)
\end{array}
\right)
=
 \left( 
\begin{array}{cc}
-\partial_{\beta} e & -\partial_{\beta} v\\
-\partial_{\lambda} e & -\partial_{\lambda} v
\end{array}
\right).
\end{equation}
Observe that $\chi$ is symmetric since $\partial_{\beta} v = \partial_{\lambda} e$. 
We define the normalized current $\widehat{J}_{x.x+1}$ as
\begin{equation}
\label{eq:normalized_current}
\widehat{ {J}}_{x,x+1}  = J_{x,x+1}  - {\mf J}(\bar \xi) - D{\mf J}({\bar \xi}) \left( \xi_x -{\bar \xi} \right), 
\qquad 
J_{x,x+1}= 
\left(
\begin{array}{c}
j_{x,x+1}^e\\
j_{x,x+1}^v
\end{array}
\right),
\end{equation}
where the term $D{\mf J}({\bar \xi}) \left( \xi_x -{\bar \xi} \right)$ has been subtracted in order to study fluctuations in the transport frame. We introduce now
\begin{equation*}
\fl
{\mc W}_{x,x+1} (t) =\int_{0}^{t} \left[ \widehat{ J}_{x,x+1} (s)  
- \gamma \nabla
\left(
\begin{array}{c}
V(\eta_x (s)\\
\eta_x (s)
\end{array}
\right) 
\right] ds 
+
M_x^\gamma(t),
\quad
M_{x}^\gamma (t)=
\left(
\begin{array}{c}
M_x^{e,\gamma} (t) \\
M_x^{v,\gamma} (t)
\end{array}
\right). 
\end{equation*}
Then, the diffusivity is defined by
\begin{equation}
\label{eq:diff}
\fl {\mc D}_{\beta,\lambda} (t)  = \lim_{N \to \infty} \frac{\chi^{-1}}{2}\frac{1} {(2N+1)t} \left\langle \left(  \sum_{|x| \le N} {\mc W}_{x,x+1} (t) \, \right) \left(  \sum_{|x| \le N} {\mc W}_{x,x+1} (t) \right)^* \right\rangle_{\beta,\lambda}. 
\end{equation}

Let us rewrite this expression in a more convenient way, introducing he current-current correlation function
\begin{equation}
\label{eq:cbl}
C_{\beta,\lambda} (t)= \sum_{x \in \ZZ} \left \langle \,  \widehat{ {J}}_{x,x+1} (t)   \, \widehat{ {J}}_{0,1} (0)^* \, \right \rangle_{\beta,\lambda}.
\end{equation}
Remark first that the terms in (\ref{eq:diff}) coming from the discrete gradients in ${\mc W}_{x,x+1}$ disappear. Indeed, after summation over $|x| \le N$, only two boundary terms at $x = \pm N$ remain and since they are divided by $2N+1$, their contribution vanishes in the limit $N \to +\infty$. Besides, the cross terms between the martingales and the normalized currents $\widehat{J}_{y,y+1}$ cancel. The argument is based on a time-reversal property of the current and can be read in~\cite{BBO2}. Moreover, by (\ref{eq:martqv}), 
\begin{equation*}
\left\langle{ \left(  \sum_{|x| \le N} M^\gamma_{x} (t)  \right) \left(  \sum_{|x| \le N} M^\gamma_{x} (t) \right)^*}\right\rangle_{\beta,\lambda} = 2t \, (2N +1) \, \gamma \, \chi. 
\end{equation*}
In conclusion, the diffusivity for the stochastic energy-volume conserving dynamics is given by 
\begin{equation}
\label{eq:diff2}
{\mc D}_{\beta,\lambda}^\gamma(t)= \chi^{-1} \int_0^{\infty} \left( 1-\frac{s}{t} \right)^+ C_{\beta,\lambda} (s) ds  + {\gamma} \, {\rm Id}_2,
\end{equation} 
where ${\rm Id}_2$ is the $2 \times 2$ identity matrix. Observe that the long time behavior of ${\mc D}_{\beta,\lambda} (t)$ is clearly driven by the long time behavior of the current-current correlation function. 

\begin{theo}
\label{th:GK}
Define the function $g \, : \, [0,\pi/2] \times (0,+\infty) \to \RR$ by
\[
g(\omega, t)= \frac{4}{\pi}\, e^{-8 \gamma t \sin^2 (\omega)} \cos^2 (\omega),
\]
and consider the infinite volume dynamics generated by ${\mc L}$, with the potential 
$V(r)=r^2 /2$, started at equilibrium under the Gibbs measure~$\mu_{\beta,\lambda}$. Then,  
\begin{equation}
\label{eq:expression_current_current_correlation}
\sum_{x \in \ZZ} \left \langle \,  \widehat{ {J}}_{x,x+1} (t)   \, \widehat{ {J}}_{0,1} (0) ^* \, \right \rangle_{\beta,\lambda} 
=
\left(
\begin{array}{cc}
\displaystyle \beta^{-2} \int_{0}^{\pi/2} g (\omega, t )\, d\omega &0\\[10pt]
0&0\\
\end{array}
\right).
\end{equation}
It follows that the only non-trivial term of this matrix is of order $t^{-1/2}$ as $t$ goes to infinity. In particular ${\mc D}_{\beta,\lambda}^\gamma(t)$ is of order ${\mc O}(t^{1/2})$. 
More precisely,
\[
{\mc D}_{\beta,\lambda}^\gamma(t) \sim \frac{8\sqrt{t}}{3\beta\sqrt{2\pi\gamma}} \,
\left( \begin{array}{cc}
\beta & 0 \\
\lambda & 0
\end{array} \right).
\]
\end{theo}

\begin{proof}
To simplify the notation we omit the indices $\beta, \lambda$. We denote by $\ll \cdot , \cdot \gg$ the semi-inner product defined on local integrable functions by 
\[
\fl \qquad \ll f, g \gg = \sum_{x \in \ZZ} \left[ \langle f (\theta_x g) \rangle - \langle f \rangle \langle g \rangle \right]
= \lim_{N \to +\infty} \frac{1}{2N+1} \sum_{|x|,|y| \leq N} \left\langle \theta_x f, \theta_y g 
\right\rangle - \langle f \rangle\langle g \rangle,
\]
and by ${\mc H}$ the Hilbert space obtained by the completion of the vector space of local functions with respect to $\ll \cdot , \cdot \gg$. 
Note that every discrete gradient, \textit{i.e.} a local function of the form
$\theta_1 f -f$, is equal to~0 in~${\mc H}$. 

We introduce the Laplace transform matrix ${\bf L}(z)$ (for $z>0$) 
of the current-current correlation function:
\[
{\bf L} (z) = \int_0^{\infty} e^{-z s} \ll   \widehat{ {J}}_{0,1} (s)   \, , \, \widehat{ {J}}_{0,1} (0) ^* \, \gg  ds 
= \ll (z -{\mc L})^{-1} \widehat{ J}_{0,1} \, , \, \widehat{ J}_{0,1}^*  \gg, 
\]
where we adopted the short notation $(z -{\mc L})^{-1} \widehat{ J}_{0,1}$ for
\begin{equation}
\left(
\begin{array}{c}
(z -{\mc L})^{-1} \widehat{J}_{0,1}^1 \\
(z -{\mc L})^{-1} \widehat{J}_{0,1}^2 
\end{array}.
\right)
\end{equation}
Since $V(r)=r^2/2$, it holds $Z(\beta, \lambda)= (\beta/2\pi)^{-1/2} e^{\lambda^2/ 2\beta}$, $v=\tau= -\lambda/\beta$ and $e= \frac{1}{2\beta} \left( \lambda^2 /\beta + 1 \right)$. It follows that 
\[
\chi= \left( 
\begin{array}{cc}
\displaystyle \frac{1}{2\beta^2} + \frac{\lambda^2}{\beta^3} 
& \quad \displaystyle - \frac{\lambda}{\beta^2}\\[10pt]
\displaystyle -\frac{\lambda}{\beta^2} &\quad  \displaystyle \frac{1}{\beta}
\end{array}
\right),
\]
\[
D{\mf J}({\bar \xi})=
\left(
\begin{array}{cc}
0 & -2 \tau \\
0 & -2
\end{array}
\right), \quad
\widehat{ J}_{0,1} = 
-\left(
\begin{array}{c}
(\eta_0 -\tau) (\eta_1 -\tau) +\tau (\eta_1 -\eta_0)\\
\eta_1 -\eta_0
\end{array}
\right)
\]
Since any discrete gradient is equal to zero in ${\mc H}$, 
the only non-zero entry of ${\bf L} (z)$ is the $(1,1)$ component 
\[
  {\bf L}^{1,1}(z) = \ll (z -{\mc L})^{-1} \widehat{ \eta}_{0,1} \, , \,   \widehat{ \eta}_{0,1}\gg,
  \]
where $\widehat{ \eta}_{x,y} = (\eta_x -\tau) (\eta_y -\tau)$.

The determination of the current-current correlation function therefore
amounts to solving the resolvent equation $(z- {\mc L}) u = \widehat{ \eta}_{0,1}$. 
Consider the vector space 
$\mathcal{V}$ spanned by the orthogonal basis 
$\{\widehat{\eta}_{x,x+k}\}_{x\in\mathbb{Z}, k \geq 1}$ (which is also the space 
spanned by the family $\widehat{\eta}_{x,y}$ for $x \neq y$).
This space is stable by $\mathcal{L}$ since 
\begin{enumerate}[(a)]
\item $\mathcal{A}\widehat{\eta}_{x,x+k} = (\theta_1 - 1)\widehat{\eta}_{x,x+k-1}-
  (\theta_1 - 1)\widehat{\eta}_{x-1,x+k}$ is the difference of two discrete gradients,
  hence is equal to~0 in~$\mathcal{H}$;
\item when $k\geq 2$, $\mathcal{S}\widehat{\eta}_{x,x+k} = (1+\theta_1)\widehat{\eta}_{x,x+k-1}
+ (1+\theta_{-1})\widehat{\eta}_{x,x+k+1} - 4\widehat{\eta}_{x,x+k}$, while
$\mathcal{S}\widehat{\eta}_{x,x+1} = \widehat{\eta}_{x-1,x+1}+\widehat{\eta}_{x,x+2}
-2\widehat{\eta}_{x,x+1}$.
\end{enumerate}
We may thus look for a solution of the form
\begin{equation}
u = \sum_{x\in\mathbb{Z}} \sum_{k\geq1} \rho_k(x) \widehat{\eta}_{x,x+k},
\end{equation}  
with the condition 
\begin{equation}
  \label{eq:normalization_condition_u}
  \sum_{x\in\mathbb{Z}} \sum_{k\geq1} |\rho_k(x)|^2 < +\infty
\end{equation}
since $u$ is sought in~$\mathcal{H}$.
In view of point~(a) above, $u$ should actually be a solution to
\[
(z- \gamma {\mc S}) u = \widehat{ \eta}_{0,1} = 
\sum_{x \in \ZZ} \sum_{k=1}^{\infty} F_{k} (x) \widehat{ \eta}_{x,x+k},
\]
with $F_{1} (x) = (z+2\gamma) \rho_1 (x) - \gamma (\rho_2 (x) + \rho_2 (x-1))$ 
and, for $k\geq2$,
\[
\fl \qquad F_{k} (x) = (z+4\gamma) \rho_{k} (x)
- \gamma \Big(\rho_{k-1} (x) + \rho_{k-1} (x+1) + \rho_{k+1} (x) + \rho_{k+1} (x-1)\Big).
\]
Using the fact that $\{\widehat{\eta}_{x,x+k}\}_{x\in\mathbb{Z}, k \geq 1}$ is an orthogonal 
basis of~$\mathcal{V}$, and identifying the coefficients in front of the different
terms, it follows that
\begin{equation}
\label{eq:64}
F_{k} (x) = {\bf 1}_{\{k=1, x=0\}}.
\end{equation}
Introducing the Fourier transform $\widehat{h}(\omega)$ (for $\omega \in \TT_1$)
of a given function $h \in l^2(\ZZ,\RR)$: 
\[
\widehat{h}(\omega) = \sum_{x \in \ZZ} e^{2 \ri \pi \omega x} h(x),
\]
the conditions~(\ref{eq:64}) can be equivalently reformulated as
\begin{equation}
\label{eq:65}
\fl \qquad \left\{ 
\eqalign{
(z+2\gamma) \widehat{\rho}_1 (\omega) - \gamma (1 +e^{2 \ri \pi \omega}) 
\widehat{\rho}_2(\omega) = 1,\cr
(z+4\gamma) \widehat{\rho}_{k}(\omega) - \gamma (1 +e^{-2 \ri \pi \omega}) 
\widehat{\rho}_{k-1}(\omega) - \gamma (1 +e^{2 \ri \pi \omega})  
\widehat{\rho}_{k+1}(\omega) = 0, \quad k \ge 2.
}\right.
\end{equation}
By Parseval's relation, condition~(\ref{eq:normalization_condition_u}) is equivalent to
\[
\sum_{k \ge 1} \int_{\TT} | \widehat{\rho}_k(\omega)|^2 \, d\omega < +\infty.
\]
It is then easy to show that (\ref{eq:65}) and the above integrability condition lead to
$\widehat{\rho}_k(\omega)= \widehat{\rho}_1(\omega)(X(\omega))^{k-1}$,
with
\[
X(\omega) = \frac{2+z/(2\gamma)}{1+e^{2\ri\pi \omega}} \left( 1- \sqrt{1- \left(\frac{\cos (\pi \omega) }{1+z/(4\gamma)}\right)^2}\right),
\]
and the boundary condition
$(z+2\gamma) \widehat{ \rho}_1 (\omega) - \gamma (1+e^{2 \ri \pi \omega}) 
\widehat{ \rho}_1 (\omega) X (\omega) = 1.$
It follows that 
\[
{\bf L}^{1,1}(z) = \ll u, \widehat{ \eta}_{0,1} \gg = \beta^{-2} \sum_{x \in \ZZ} \rho_{1} (x) 
= \beta^{-2} \widehat{ \rho}_1(0)= \frac{1}{2\gamma\beta^2}\, {\mc T}\left(\frac{z}{4\gamma}
\right), 
\]
where
\[
{\mc T} (y) = \left[y + \sqrt{(1+y)^2- 1 }\right]^{-1}.
\]
A simple computation shows that, for any $y>0$,
\begin{equation*}
\fl \qquad \int_0^{\infty} e^{-y t} \left[ \frac{2}{\pi} \int_0^{\pi/2} e^{-2t\sin^2(x)} \cos^{2} (x) dx \right] dt = \frac{2}{\pi} \int_0^{\pi/2} \frac{\cos^2 (x)}{y +2 \sin^{2} (x)} \, dx = {\mc T} (y).
\end{equation*}
Since the Laplace transform uniquely characterizes the underlying function, 
we deduce that the current-current correlation function is indeed given 
by~(\ref{eq:expression_current_current_correlation}).

Now, by dominated convergence, the following limit holds as $t \to +\infty$:
\begin{eqnarray*}
\sqrt{t} \int_0^{\pi/2} g(\omega,t) \, d\omega & = \frac{4}{\pi} \int_0^{\pi\sqrt{t}/2} 
\exp\left(-8\gamma t \sin^2\left(\frac{\omega}{\sqrt{t}}\right) \right) 
\cos^2\left(\frac{\omega}{\sqrt{t}}\right) \, d\omega \cr
& \longrightarrow \frac{4}{\pi} \int_0^{+\infty} e^{-8\gamma \omega^2} \, d\omega 
= \frac{1}{\sqrt{2\pi\gamma}}.
\end{eqnarray*}
We then obtain the desired result with~(\ref{eq:diff2}).
\end{proof}

%------------------- NUMERIQUE -------------------

\section{Steady-state nonequilibrium systems}
\label{sec:sim}

\subsection{General setting}

The results of the previous section were limited to harmonic potentials (and the 
Kac-van-Moerbecke potential in the special case $\gamma = 0$).
For generic anharmonic potentials, we can only provide numerical evidence of the super-diffusivity.
However, it is difficult to estimate numerically the time autocorrelation functions of the 
currents because of their expected long-time tails, and because statistical errors are
very large (in relative value) when $t$ is large. 
Also, for finite systems (the only ones we can simulate on a computer), 
the autocorrelation is generically exponentially decreasing for anharmonic potentials,
and, to obtain meaningful results, the thermodynamic limit should be taken before 
the long-time limit.

A more tenable approach consists in studying a nonequilibrium system in its steady-state.
We consider a finite system of length $2N+1$ in contact with two 
thermostats which fix the value of the energy at the boundaries. 
The generator of the dynamics is given by
\begin{equation}
  \label{eq:generator_open}
        {\mc L}_{N,{\rm open}} = {\mc A}_{N} + \gamma {\mc S}_N + \lambda_\ell {\mc B}_{-N,T_\ell} + 
        \lambda_{r}  {\mc B}_{N,T_r},
\end{equation}
where ${\mc A}_N$ (resp. ${\mc S}_N$) is defined  by (\ref{eq:A}) (resp. (\ref{eq:S})) 
with $\Lambda_N=\{-N,\ldots,N\}$ and 
${\mc B}_{x,T} = T \partial_{\eta_x}^2 -V' (\eta_x) \partial_{\eta_x}$. 
The positive parameters $\lambda_\ell$ and $\lambda_r$ are the intensities of the thermostats. 

The generator ${\mc B}_{x,T}$ can be seen as a thermostatting mechanism since the semigroup $(S_t )_{t \ge 0}$ generated by ${\mc B}_{x,T}$ has, under suitable assumptions on $V$, a unique (reversible) invariant probability measure on $\RR$ given by
\begin{equation}
\nu_T(d\xi) = Z(\beta,0) e^{-\beta V(\xi)} \, d\xi , \qquad \beta = T^{-1}, 
\end{equation}  
and $S_t f $ converges exponentially fast to $\nu_T (f)$ 
for any observable $f \in L^2(\nu_T)$ (we will however not use these facts in the sequel).
Observe that 
\[
\nu_T ( V)= -\partial_{\beta} (\log Z(\beta,0))= e(\beta,0).
\] 
Hence, in order to fix the energy at site $-N$ (resp. $N$) to the value $e_\ell$ (resp. $e_r$), we have to choose $\beta_\ell= T_{\ell}^{-1}$ (resp. $\beta_r = T_r^{-1}$) such that $e(\beta_\ell,0)=e_\ell$ (resp. $e(\beta_r, 0)=e_r$).

In the special case when $T_\ell = T_r$, the Gibbs 
measure~$\mu_{\beta,0}|_{\Lambda_N}$ is invariant, and it can then be shown
(see \ref{sec:app_A}) that the the law of the stochastic 
process associated to~(\ref{eq:generator_open}) converges exponentially fast
to~$\mu_{\beta,0}|_{\Lambda_N}$. 

The proof of the existence and uniqueness of an invariant measure in the case
when $T_\ell \neq T_r$ is given in Proposition~\ref{prop:rb} for a suitable class of potentials
(see \ref{sec:app_A} for a proof in the case $\gamma = 0$ and \ref{sec:app_B} 
for a proof when $\gamma > 0$). 
We denote by $\langle \cdot \rangle_{\rm ss}$ the unique stationary state for the 
dynamics generated by ${\mc L}_{N,{\rm open}}$.

\begin{prop}
\label{prop:rb}
Assume that the smooth potential $V$ satisfies
\begin{itemize}
\item \emph{[Growth at infinity]} there exist real 
  constants $k \ge 2$, $a_k >0$, $C>0$ such that
\begin{equation}
\label{eq:potpot0}
\eqalign{
\lim_{\lambda \to + \infty} \lambda^{-k} V (\lambda q) = a_k |q|^k, \cr
\lim_{\lambda \to + \infty} \lambda^{1-k} V' (\lambda q) = k a_k |q|^{k-1} {\rm sign} (q),
}
\end{equation}
\begin{equation}
\label{eq:potpot}
\lim_{|q| \to \infty} \frac{V'' (q)}{(V' (q))^2} =0;  
\end{equation}
\item \emph{[Non-degeneracy]} For any $q \in \RR$ there exists $m:=m(q) \ge 2$ such that
 \begin{equation}
 \label{eq:nd}
 V^{(m)} (q) \ne 0.
 \end{equation}
\end{itemize}
Then, there exists a unique stationary probability measure for 
the Markov process generated by ${\mc L}_{N,{\rm open}}$. This stationary 
state has a smooth positive density with respect to the Lebesgue measure. 
\end{prop}

The energy currents $j^{e,\gamma}_{x,x+1}$, which are such that
${\mc L}_{N,{\rm open}} (V (\eta_x)) = -\nabla j_{x-1,x}^e$ (for $x=-N, \ldots,N+1$), 
are given by the expressions~(\ref{eq:current_gamma})
for $x = -N+1,\dots,N-1$ while
\[
\eqalign{
j^{e,\gamma}_{-N-1,-N} = \lambda_\ell \left[ T_\ell V'' (\eta_{-N}) -(V' (\eta_{-N}))^2 \right],\\
j^{e,\gamma}_{N,N+1}   = -\lambda_r \left[ T_r V'' (\eta_{N}) -(V' (\eta_{N}))^2 \right].
}
\]
Since $\langle {\mc L}_{N,{\rm open}} (V (\eta_x)) \rangle_{\rm ss} = 0$,
it follows that, for any $x = -N,\ldots,N+1$, 
$\langle j_{x,x+1}^{e,\gamma} \rangle_{\rm ss}$ is equal to a 
constant $J_N^\gamma(T_\ell,T_r)$ independent of $x$. In fact,
\begin{equation}
\label{eq:total_current_NESS}
J^{\gamma}_N(T_\ell,T_r) = \left\langle \mathcal{J}_N^\gamma 
\right\rangle_{\rm ss}, 
\qquad
\mathcal{J}_N^\gamma = \frac{1}{2N}\sum_{x=-N}^{N-1} j_{x,x+1}^{e,\gamma}.
\end{equation} 
The latter equation is interesting from a numerical viewpoint since it allows to 
perform some spatial averaging, hence reducing the statistical error of the results.
Finally, let us mention that, for \emph{finite systems}, 
standard results of linear response theory
(see for instance~\cite{BO,RB}) allow to relate the nonequilibrium current $J_N^\gamma$
to the current autocorrelation at equilibrium as
\[
\lim_{T_\ell, T_r \to T} 
\frac{J^{\gamma}_N(T_\ell,T_r)}{T_\ell-T_r} = \frac{2N}{T^2} \int_0^{+\infty} 
\mathbb{E}_{\mu_{T^{-1},0}} \left[
\mathcal{J}_N^\gamma(t) \mathcal{J}_N^\gamma(0) \right ] dt,
\]
where the expectation is over all initial conditions and realizations of the paths.

\subsection{Harmonic potentials}
\label{sec:harmonic_pot_NESS}

The current can be computed explicitly for harmonic potentials when no stochastic perturbation
is present. In this case, a ballistic behavior (no damping of the current with the system size)
is found, as for standard harmonic oscillator chains~\cite{RLL}. It can also be shown that the macroscopic profile of energy $q \in [-1,1] \to \langle V(\eta_{[xq/N]}) \rangle_{\rm ss}$ is flat and equal to $(T_\ell +T_r)/4$.

\begin{prop}
  Consider the system described by the 
  generator~(\ref{eq:generator_open}) with $\gamma = 0$, for the harmonic potential $V(r)=r^2/2$.
  Then the steady-state energy current is 
  \begin{equation}
    \label{eq:harmonic_current}
    J^0_N(T_\ell,T_r) = \frac{T_\ell-T_r}{\lambda_\ell+\lambda_\ell^{-1}+\lambda_r+\lambda_r^{-1}}.
  \end{equation}
\end{prop}

\begin{proof}
  Introduce $C_{x,y} = \langle \eta_x \eta_y \rangle_{\rm ss} = C_{y,x}$
  for $(x,y) \in \Lambda^2$.
  First, note that the stationarity of the current 
  implies $C_{x,x+1} = -J_N^0(T_\ell,T_r)$ for $x=-N,\ldots,N$
  and 
  \begin{equation}
    \label{eq:current_boundaries}
    \lambda_\ell (T_\ell-C_{-N,-N}) = -\lambda_r(T_r-C_{N,N}) = J_N^0(T_\ell,T_r).
  \end{equation}
  Besides, for $x = -N+1,\ldots,N-2$, it holds
  \[
  \langle {\mc L}_{N,{\rm open}} (\eta_x\eta_{x+1}) \rangle_{\rm ss} = 0 = C_{x+1,x+1}+C_{x,x+2}
  -(C_{x,x}+C_{x-1,x+1}).
  \]
  Therefore, $C_{N-1,N-1}+C_{N-2,N} = C_{-N+1,-N+1}+C_{-N,-N+2}$.
  Now, 
  \[
  \fl \qquad \langle {\mc L}_{N,{\rm open}} (\eta_{-N}\eta_{-N+1}) \rangle_{\rm ss}= 0
  = C_{-N+1,-N+1} + C_{-N+2,N} - \lambda_\ell C_{-N,-N+1} - C_{-N,-N},
  \]
  and, similarly, $C_{N,N-2} + C_{N-1,N-1} = C_{N,N} - \lambda_r C_{N,N-1}$.
  In conclusion,
  \[
  \lambda_\ell C_{-N,-N+1} + C_{-N,-N} = C_{N,N} - \lambda_r C_{N,N-1}.
  \]
  The expression~(\ref{eq:harmonic_current}) is then obtained by combining the above equality 
  with~(\ref{eq:current_boundaries}) and the relation 
  $C_{N,N-1} =  C_{-N,-N+1} = -J_N^0(T_\ell,T_r)$.
\end{proof}
 
When $\gamma > 0$, we expect to observe an 
anomalous diffusion with divergence rate $\delta=1/2$, \textit{i.e.} $N J^\gamma_N \sim C_\gamma 
\sqrt{N}$ when $N$~is large enough. 
This is indeed confirmed by numerical simulations, see Figure~\ref{fig:harmonic}.

\subsection{Anharmonic potentials}

The nonlinear case is much more difficult. We estimate the 
exponent $\delta \ge 0$ such that 
\begin{equation}
  \label{eq:def_delta}
  N J_N^{\gamma} \sim N^{\delta}
\end{equation}
using numerical simulations. 
If $\delta=0$, the system is a normal conductor of energy.
If on the other hand $\delta > 0$, it is a superconductor.

\subsubsection{Numerical scheme.}
The time-discretization of the dynamics with generator~(\ref{eq:generator_open})
is done with a standard splitting strategy, decomposing the generator as the sum of
a deterministic part ${\mc A}_{N}$, a thermostat part $\lambda_\ell {\mc B}_{-N,T_\ell} + 
\lambda_{r} {\mc B}_{N,T_r}$, and the stochastic perturbation
$\gamma {\mc S}_N$, and integrating each part in this order.
We denote by $\Delta t$ the time-step. 

A simple numerical scheme for the deterministic part of the dynamics
by relying on the Hamiltonian interpretation of the system.
The longtime integration 
of Hamiltonian system is well understood. The most standard scheme used in practice
is the so-called St\"ormer-Verlet scheme (see again~\cite{HLW}), which,
for separable Hamiltonians (where the total energy is the sum of a kinetic part depending
only on the momenta, and a potential part depending only on the positions) 
can be seen as a Strang approximation of the Hamiltonian evolution with positions and momenta
updated successively.
This amounts here to updating successively variables with odd and even indices.
This corresponds to the Strang 
splitting based on the following decomposition of the generator:
\[
\mathcal{A}_N = \mathcal{A}_N^{\rm even} + \mathcal{A}_N^{\rm odd},  
\]
with, in the case when $N$ is even,
\[
\fl \mathcal{A}_N^{\rm even} = \sum_{x = 1}^{N-1} 
\Big( V'(\eta_{2x+1-N})-V'(\eta_{2x-1-N})\Big) \partial_{\eta_{2x-N}}
+ V'(\eta_{-N+1}) \partial_{\eta_{-N}} - V'(\eta_{N-1}) \partial_{\eta_{N}},
\]
and a similar definition for $\mathcal{A}_N^{\rm odd}$.
This splitting is particularly convenient since the time 
evolutions generated by $\mathcal{A}_N^{\rm even}$ are 
$\mathcal{A}_N^{\rm odd}$ are both analytically integrable.
In conclusion, the numerical scheme used for the deterministic part reads:
\[
\fl \quad \left \{
\eqalign{
  \eta_x^{n+1/2} 
  = \eta_x^n + \frac{\Delta t}{2} \, \Big( V'(\eta^n_{x+1})-V'(\eta^n_{x-1}) \Big), 
  \quad x = -N+1,-N+3,\dots,N-1 \cr
  \eta_y^{n+1} = \eta_y^{n} + \Delta t \, 
  \Big( V'(\eta^{n+1/2}_{y+1})-V'(\eta^{n+1/2}_{y-1}) \Big), 
  \quad y = -N,-N+2,\dots,N \cr
  \eta_x^{n+1} 
  = \eta_x^{n+1/2} + \frac{\Delta t}{2} \, \Big( V'(\eta^{n+1}_{x+1})-V'(\eta^{n+1}_{x-1}) \Big), 
  \quad x = -N+1,\dots,N-1.
}
\right.
\]

The thermostat part is taken care of by a simple Euler-Maruyama discretization:
\[
\left\{
\eqalign{
\eta_{-N}^{n+1} = \eta_{-N}^n - \lambda_\ell
\Delta t \, V'(\eta_{-N}^n) + \sqrt{2 \lambda_\ell T_\ell \Delta t} \, G_{-N}^n,\cr
\eta_{N}^{n+1} = \eta_{N}^n - \lambda_r
\Delta t \, V'(\eta_{N}^n) + \sqrt{2 \lambda_r T_r \Delta t} \, G_{N}^n,
}\right.
\]
where $(G_{\pm N}^n)$ are independent and identically distributed standard Gaussian
random variables.
Finally, the noise term with generator $\gamma \mathcal{S}$ is simulated by 
exchanging $\eta_x$ and~$\eta_{x+1}$ (for $x=-N,\ldots,N-1$) 
at exponentially distributed random times,
with an average time $\gamma^{-1}$ between two such exchanges. More precisely,
we attach to each couple $(x,x+1)$ a random time~$\tau_x^m$, with $\tau_i^0$ 
drawn from an exponential 
law with parameter $\gamma$, and where $(\tau_i^m)_{i,m}$ are independent. 
This time is updated as follows:
If $\tau_x^{m} \geq \Delta t$, then $\tau_x^{m+1} = \tau_x^{m} - \Delta t$, 
otherwise~$\eta_x$ and $\eta_{x+1}$ are exchanged and a new exchange time $\tau_x^{m+1}$ is
resampled from an exponential law of parameter~$\gamma$.  

\subsubsection{Numerical results.}
We considered the following potentials:
\begin{enumerate}
\item harmonic potential $V(r) = r^2/2$;
\item anharmonic FPU-like potential $V(r) = r^2/2+r^4/4$;
\item Kac-van-Moerbecke potential $V(r) = e^{-r}+r-1$;
\item rotor $V(r) = 1 - \cos(r)$.
\end{enumerate}
The time-step $\Delta t$ is chosen to ensure a good longtime preservation of energy
for the deterministic dynamics in the absence of stochastic perturbation.
In order to have a relative error in energy less than~$10^{-5}$,
we used $\Delta t = 0.005$ except for rotors where $\Delta t = 0.025$.
We set $\lambda_\ell = \lambda_r = 1$ (this choice maximizes the observed current in
the harmonic case according to~(\ref{eq:harmonic_current})), 
and considered a small temperature difference $T_\ell = 1.1$, $T_r = 0.9$.
We performed $10^9$ iterations for $N\leq 2^{12}$, $5 \times 10^8$ iterations for
$N = 2^{13}$, $2.5 \times 10^8$ iterations for $N = 2^{14}$ and $1.25 \times 10^8$ iterations for
$N = 2^{15}$. The corresponding system size in the latter case is $2N+1 = 65,537$.

\begin{figure}
\begin{center}
\includegraphics[width=12.cm]{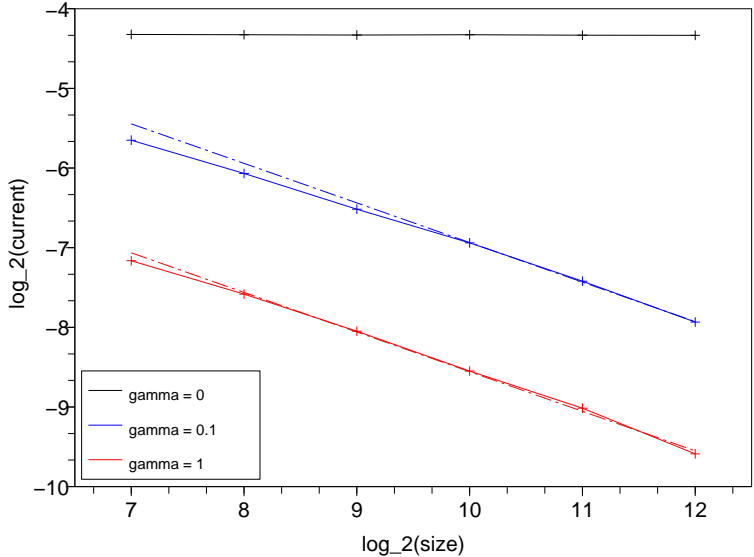}
\caption{\label{fig:harmonic}
  Current $J_N^\gamma$ as a function of the system size $2N+1$ for harmonic potentials
  $V(r) = r^2/2$ (in $\log_2-\log_2$ scale), with three values of $\gamma$: $\gamma=0$ (black), 
  $\gamma = 0.1$ (blue) and $\gamma = 1$ (red). 
}
\end{center}
\end{figure}
\begin{figure}
\begin{center}
\includegraphics[width=12.cm]{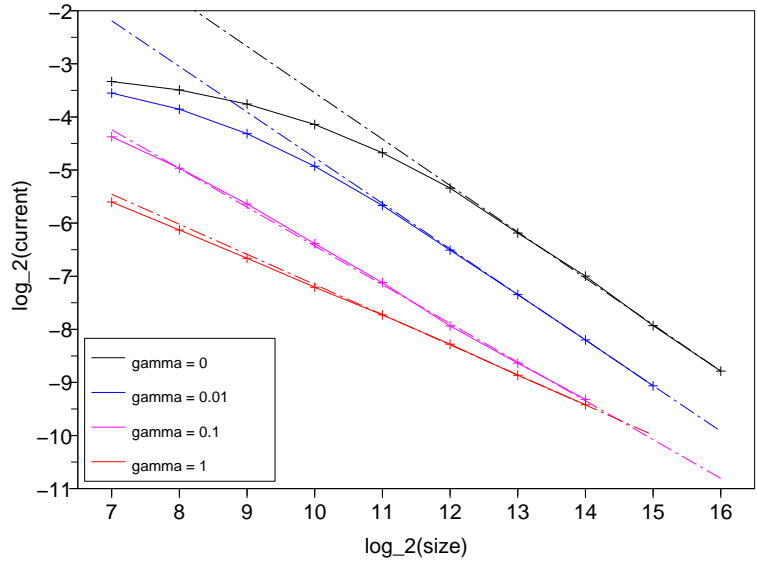}
\caption{\label{fig:R4}
  Current as a function of the system size $2N+1$ for anharmonic potentials
  $V(r) = r^2/2 + r^4/4$ (in $\log_2-\log_2$ scale), 
  with four values of $\gamma$: $\gamma=0$ (black), 
  $\gamma = 0.01$ (blue), $\gamma = 0.1$ (pink) and $\gamma = 1$ (red).
}
\end{center}
\end{figure}
\begin{figure}
\begin{center}
\includegraphics[width=12.cm]{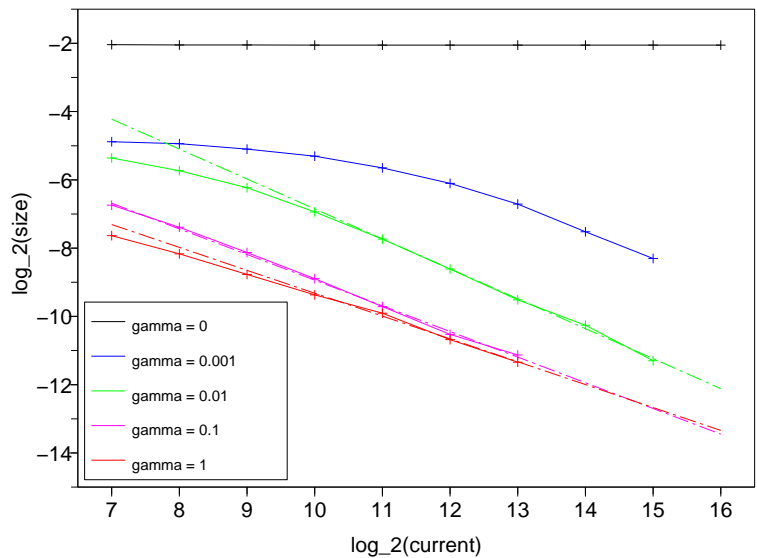}
\caption{\label{fig:Toda}
  Current as a function of the system size $2N+1$ for the Kac-van-Moerbecke
  potential $V(r) = e^{-r}+r-1$ (in $\log_2-\log_2$ scale), 
  with five values of $\gamma$: $\gamma=0$ (black), 
  $\gamma = 0.001$ (blue), $\gamma = 0.01$ (green), 
  $\gamma = 0.1$ (pink) and $\gamma = 1$ (red).
}
\end{center}
\end{figure}
\begin{figure}
\begin{center}
\includegraphics[width=12.cm]{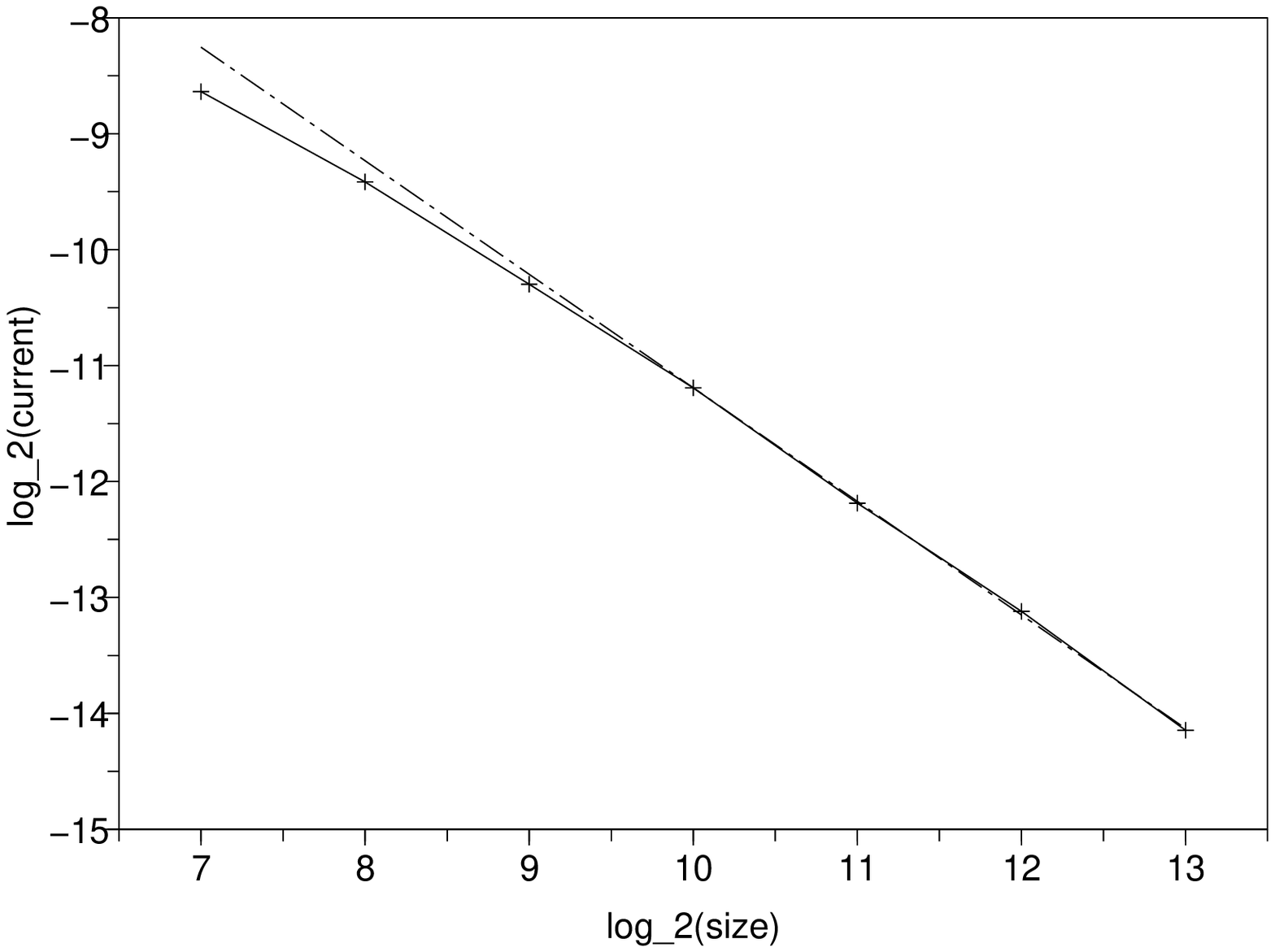}
\caption{\label{fig:rotor} Current as a function of the system size $2N+1$ for
  a chain of rotors: $V(r) = 1 - \cos(r)$. A normal conductivity is observed
  since the estimated conductivity exponent is $\delta \simeq 0.02$.}
\end{center}
\end{figure}

The simulation results are presented in Figures~\ref{fig:harmonic} to~\ref{fig:rotor}.
The conductivity exponents extracted from the numerical simulations
presented in Figures~\ref{fig:harmonic} to~\ref{fig:Toda} 
are reported in Table~\ref{tab:exponents}. 
Exponents in the harmonic case agree with their expected values
(see the discussion in Section~\ref{sec:harmonic_pot_NESS}).
For nonlinear potentials, except for the singular value $\delta = 1$ when 
$\gamma = 0$, the exponents seem to be monotonically increasing with $\gamma$.
We conjecture that these exponents should attain the limiting value~0.5 as $\gamma \to +\infty$.
A similar behavior of the exponents is observed for Toda chains~\cite{I}.
Note also that the value found for $\gamma = 0$ with the anharmonic FPU potential 
$V(r) = r^2/2 + r^4/4$ is smaller than the corresponding value for standard
oscillator chains, which is around 0.33 (see~\cite{MDN07}).
\begin{table}
\caption{\label{tab:exponents} Conductivity exponents $\delta$ (see~(\ref{eq:def_delta}) for
  the definition of~$\delta$).} 
\begin{indented}
\item[]\begin{tabular}{@{}llll}
\br
$\mathbf{\gamma}$ & harmonic & anharmonic & KVM \\
\mr
$0$               & 1     & 0.13 & 1    \\
$0.01$\phantom{0} & --    & 0.14 & 0.12 \\
$0.1$\phantom{00} & 0.50  & 0.27 & 0.25 \\
$1$\phantom{0.00} & 0.50  & 0.43 & 0.33 \\ 
\hline
\end{tabular}
\end{indented}
\end{table}

%----------- APPENDICE --------------------

\appendix

\section{Existence and uniqueness of the stationary state for $\gamma=0$}
\label{sec:app_A}

We adapt in our context the methods introduced in~\cite{EPRB} 
(see~\cite{RB} for a review) for the chains of coupled oscillators, 
following closely the exposition given in~\cite{C}. 
The proof is divided in 3~steps: (i) the aim of \ref{sec:existence_pbm} is to prove 
the Lyapunov condition~(\ref{eq:TTT});
(ii) we then prove the smoothness of the transition probability  
using hypoellipticity arguments (\ref{sec:smoothness_pbm}), and finally (iii) show that the process
is irreducible (\ref{sec:uniqueness_pbm}). 
These three arguments, together with~\cite[Theorem~8.9]{RBa} 
allow to conclude the existence and the uniqueness of an invariant measure
with smooth density, as well as the exponential convergence of the law of the process
to the invariant measure.

\subsection{Definition of the process}

The process generated by ${\mc L}_{N,{\rm open}}$ is 
denoted by $(\eta (t))_{t \ge 0}$ and is the solution of the system 
\[
\rm{(S)}\ \left\{ 
\eqalign{
d\eta_{-N} (t) = V' (\eta_{-N+1}) dt - \lambda_\ell V' (\eta_{-N}) dt + 
\sqrt{2 \lambda_\ell T_\ell } \, dB_{-N} (t),\cr
d \eta_{x} (t) = (V'(\eta_{x+1}) -V' (\eta_{x-1})) dt, \qquad x=-N+1,\ldots,N-1,\cr
d \eta_{N} (t) = -V' (\eta_{N-1}) dt - \lambda_r V' (\eta_{N}) dt 
+ \sqrt{2 \lambda_r T_r } \, dB_{N} (t),
} \right.
\] 
where $B_{-N}, B_N$ are two standard independent Brownian motions. 
Since the drift coefficients are only locally Lipschitz, it is unclear 
whether a solution to $({\rm S})$ exists for any time. 
This problem can however easily be solved by using a suitable Lyapunov function,
as we now show.
For $\alpha>0$, consider
\begin{equation}
  \label{eq:Lyapunov_fct}
  W_{\alpha} (\eta) = \exp \left\{ \alpha \sum_{x=-N}^N V(\eta_x) \right\}.
\end{equation}  
A simple computation shows that
\begin{eqnarray*}
({\mc L}_{N,{\rm open}} W_{\alpha}) (\eta) = \alpha & \Big[ \lambda_r T_r V'' (\eta_N) 
+\lambda_r (T_r \alpha -1) (V' (\eta_N))^2 \cr
& + \lambda_\ell T_\ell V'' (\eta_{-N}) +\lambda_\ell (T_\ell \alpha -1) 
(V' (\eta_{-N}))^2\Big ] W_{\alpha},
\end{eqnarray*}
so that, using the assumption (\ref{eq:potpot}), 
${\mc L}_{N,{\rm open}} W_{\alpha} \leq A W_\alpha$ 
for some positive constant $A$ as soon as $\alpha$ is sufficiently small.
% decomposer le comportement a l'infini, et ce qui se passe pres de l'origine
This is sufficient to prove the existence of a unique solution to $({\rm S})$ 
for any initial condition and for any time $t\ge 0$. We denote by~$(T_t)_{t \ge 0}$ 
the corresponding semigroup.

\subsection{Proof of the Lyapunov condition~(\ref{eq:TTT})}
\label{sec:existence_pbm}

Given a solution $(\eta (t))_{t \geq 0}$ 
of~$\rm{(S)}$, we define the scaled process $\eta^E$ (for $E>0$) 
as follows: 
\begin{equation}
  \label{eq:scaled_process}
  \eta^E (t) = E^{-1/k} \eta (t E^{2/k -1}).
\end{equation}
The scaling is chosen so that there is no explicit dependence in~$E$ in the formal limit
$E \to +\infty$ (see the system~${\rm (S_\infty)}$ below).
The process $\eta^E$ is a solution of
\begin{equation}
\fl \rm{(S_E)}\, \left\{
\eqalign{
d\eta^E_{-N} (t) =E^{1/k -1} \Big( V' (E^{1/k} \eta^E_{-N+1}) - \lambda_\ell V' (E^{1/k} \eta^E_{-N}) \Big) dt \cr
\qquad \qquad + \sqrt{\frac{2 \lambda_\ell T_\ell}{E} } \, dB_{-N} (t),\cr
d \eta^E_{x} (t) = E^{1/k -1} \Big(V'(E^{1/k} \eta^E_{x+1}) -V' (E^{1/k} \eta^E_{x-1})\Big) dt, \cr
\qquad \qquad x=-N+1, \ldots,N-1,\cr
d\eta^E_{N} (t) =-E^{1/k -1} \Big( V' (E^{1/k} \eta^E_{N-1}) + 
\lambda_r V' (E^{1/k} \eta^E_{N}) \Big) dt \cr
\qquad \qquad + \sqrt{\frac{2 \lambda_r T_r}{E} } \, dB_{N} (t),
} \right.
\end{equation}
where $B_{-N}, B_N$ are two standard independent Brownian motions. Observe that the scaled process is such that if $H (\eta (0))=E$ then $H_{E} (\eta^{E} (0)) =1$, where the scaled energy 
$H_E$ is
\begin{equation*}
H_E (\eta) = \frac{1}{E} \sum_{x=-N}^N V (E^{1/k} \eta_x).
\end{equation*}
In the limit $E \to +\infty$, the noise disappears 
and, by the scaling property~(\ref{eq:potpot0}) of the potential, 
$\rm{(S_E)}$ reduces formally to the following deterministic system:
\begin{equation*}
\fl \qquad \rm{(S_{\infty})}\; \left\{
\eqalign{
d\eta_{-N} (t) = \Big( \theta' (\eta_{-N+1}) - \lambda_\ell \theta' (\eta_{-N}) \Big) dt,\cr
d \eta_{x} (t) = \Big(\theta' (\eta_{x+1}) -\theta' (\eta_{x-1}) \Big) dt, 
\qquad \qquad x=-N+1, \ldots,N-1,\cr
d \eta_{N} (t) = - \Big( \theta' (\eta_{N-1}) + \lambda_r \theta' (\eta_{N}) \Big) dt, 
} \right.
\end{equation*}
where $\theta (q) = a_k |q|^{k}$ is a $C^1$ function since $k\geq2$. 
We also introduce the corresponding limiting energy function
\begin{equation*}
H_{\infty} (\eta) = \sum_{x=-N}^N \theta (\eta_x).
\end{equation*}

Arguing by contradiction, 
it is easy to show that if $({\widetilde \eta} (t))_{t \geq 0}$ 
is a solution of $(\rm{S_{\infty}})$ starting from an initial condition
${\widetilde \eta} (0)$ such that  $H_{\infty} ({\widetilde \eta} (0)) =1$, then, for any $\tau>0$, 
\begin{equation}
\label{eq:diss}
\int_0^\tau \left[ (\theta' ({\widetilde \eta}_{-N} (s)))^2 + 
(\theta' ({\widetilde \eta}_{N} (s)))^2 \right] \, ds  \; > 0. 
\end{equation} 
By using the continuity of solutions of stochastic differential equations with respect to both parameters and starting points, we can then state the following asymptotic result.

\begin{lemma}
\label{lem:th}
Assume that $((\eta^n(t))_{t \ge 0})_{n \in \mathbb{N}}$ is a sequence 
of solutions of $({\rm S})$ starting 
from $\eta^n(0)$, with $E_n = H(\eta^n(0)) \to + \infty$. 
Then there exists a subsequence $((\eta^{m}(t))_{t \geq 0})_{m\in \mathbb{N} }$ 
such that, for any $C>0$ and $t_0 >0$, 
\begin{equation*}
\lim_{m \to + \infty} {\mathbb E}_{\eta^{m}(0)} \left[ \exp \left( -C \int_0^{t_0}  \left[ (\theta ' (\eta_{-N}^m (s) )^2+ (\theta ' (\eta_N^m (s) )^2 \right] \, ds \right)\right] =0.
\end{equation*} 
\end{lemma}

\begin{proof}
Recall that $(\eta^{n,E_n} (t))_{t \ge 0}$ is the solution of $({\rm S}_{E_n})$ starting from 
$\eta^{n,E_n}(0)$, which is such that $H_{E_n} (\eta^{n,E_n}(0))=1$. 
This implies that the sequence $( \eta^{n,E_n}(0))_{n \in \mathbb{N}}$ 
remains in a compact set and we can extract a subsequence, denoted by 
$(\eta^{m,E_{m}}(0))_{m \in \mathbb{N}}$, such that $\eta^{m,E_{m}}(0)$ converges to some
element ${\widetilde \eta}^* \in \mathbb{R}^{\{-N, \ldots, N\}}$. 
By continuity, $H_{\infty} ({\widetilde \eta}^*)=1$. 
Let $({\widetilde \eta}(t))_{t \ge 0}$ be the solution to $({\rm S}_{\infty})$
starting from ${\widetilde \eta}^*$.

We fix $C > 0$ and $t_0 > 0$.
For any $\tau>0$, thanks to the continuity of solutions of 
stochastic differential equations with respect to both parameters and starting points, it holds
\begin{equation*}
\lim_{m \to \infty} {\EE} \left[ \sup_{t \le \tau} \left| {\eta^{m,E_m}}(t) -{\widetilde \eta}(t) 
\right|^2 \right] = 0.
\end{equation*}
For any $\eta \in \RR^{\{-N, \ldots, N\}}$, we denote 
the term $\left[ (\theta ' (\eta_{-N}))^2+ (\theta ' (\eta_N))^2 \right] $ by $K(\eta)$.
Then, for any $a>0$,
\begin{equation*}
\lim_{m \to \infty} \EE_{\eta^{m,E_{m}}(0) } \left[ \exp\left(-a \int_{0}^\tau K \left(\eta^{m,E_{m}}(s) 
\right ) ds\right) \right] = e^{-a I(\tau)},
\end{equation*}
where $I(\tau)= \int_{0}^\tau K \left ({\widetilde \eta}(s) \right) ds > 0$ by (\ref{eq:diss}).
Observe also that, for $E_m \geq 1$, 
\[
\fl \qquad \int_0^{t_0} K (\eta^m(s)) \, ds =  E_m \int_{0}^{E_m^{1-2/k} t_0} K(\eta^{m,E_m} (u)) \, du
\geq E_m \int_{0}^{t_0} K(\eta^{m,E_m} (u)) \, du.
\]
Now, consider some arbitrary constant $A > 0$.
By choosing $m$ sufficiently large so that $A \leq C E_{m}$, it follows
\begin{equation*}
\lim_{m \to + \infty} {\mathbb E}_{\eta^{m}(0)} \left[ \exp\left(-C \int_0^{t_0} 
K(\eta^m(s)) ds \right)\right] \le e^{-A I(t_0)}.
\end{equation*}
The result is then obtained by letting $A$ go to infinity. 
\end{proof}

We now claim that there exist a time $t_0 >0$, finite constants $b_n$ 
and ${\kappa_n} \in (0,1)$ with $\lim_{n \to \infty} \kappa_n = 0$, 
and compact sets $K_n$ such that
\begin{equation}
  \label{eq:TTT}
  \forall \eta \in \RR^{\{-N, \ldots,N\}}, 
  \qquad T_{t_0} W (\eta) \le {\kappa}_n W (\eta) + b_n {\bf 1}_{K_n} (\eta),
\end{equation}
where $W:=W_{\alpha}$ is the Lyapunov function~(\ref{eq:Lyapunov_fct}). 
By choosing compact sets ${K_n}$ of the form 
$\{ \eta \in \RR^{\{-N, \ldots,N\}}\, | \, W (\eta) \le a_n\}$ 
with $a_{n} \to + \infty$, it is enough to show that 
\[
\lim_{n\to+\infty} \sup_{\eta \not \in K_n} \frac{T_{t_0} W (\eta)}{W (\eta)} = 0.
\]
Therefore, (\ref{eq:TTT}) is a consequence of Lemma~\ref{lem:th} and the following 
result.

\begin{lemma}
If $\alpha$ is sufficiently small, there exist $c,C>0$ 
and $q>1$ such that, for all $\eta(0) \in \RR^{\{-N, \ldots,N\}}$, 
\begin{equation*}
\fl \qquad \frac{T_{t_0} W (\eta (0))}{W (\eta (0) ) } \le e^{c t_0}  
\, {\mathbb E}_{\eta(0) } \left[ \exp \left( -C \int_0^{t_0}  
\left[ (\theta ' (\eta_{-N} (s) ))^2+ (\theta ' (\eta_N (s) ))^2 \right] \, ds \right)
\right] ^{1/q}. 
\end{equation*}
\end{lemma}

\begin{proof}
It holds
\[
H(\eta (t))= H (\eta (0)) +\int_0^t {\mc L}_{N,{\rm open}} H(\eta (s)) \, ds + M (t),
\]
where $M(t)$ is a continuous martingale of quadratic variation 
\begin{eqnarray*}
\left[ M \right]_t & = \int_{0}^t \Big( {\mc L}_{N,{\rm open}}(H^2) -2 H {\mc L}_{N,{\rm open}} 
H \Big) (\eta (s)) \, ds \cr
& = 2\int_{0}^t \Big( \lambda_\ell T_\ell [V' (\eta_{-N}(s))]^2 
+ \lambda_r T_r  [V' (\eta_N (s)) ]^2 \Big) ds.
\end{eqnarray*} 
Now, for any constants $p,q \geq 1$ such that $p^{-1} +q^{-1}=1$, 
\[
\eqalign{
\fl \frac{T_t W (\eta (0) )}{W(\eta (0) )}
= {\mathbb E}_{\eta(0)} \left[ e^{ \alpha (H(\eta (t)) -H(\eta (0))}\right] \cr
\fl \qquad = {\mathbb E}_{\eta(0)} \left[ \exp\left(\alpha M (t) + 
\alpha \int_0^{t} ({\mc L}_{N,{\rm open}} H)(\eta (s)) ds \right) \right]\cr
\fl \qquad =  {\mathbb E}_{\eta(0)}\left[ \exp\left(\alpha M (t) - p \frac{\alpha^2}{2} 
\left[ M \right]_t + p \frac{\alpha^2}{2} \left[ M \right]_t +  
\alpha \int_0^{t} ({\mc L}_{N,{\rm open}} H)(\eta (s)) ds\right) \right]\cr
\fl \qquad = {\mathbb E}_{\eta(0)} \left[ X_t Y_t  \right] 
\le {\mathbb E}_{\eta(0)} [X_t^p]^{1/p} {\mathbb E}_{\eta(0)} \left[ Y_t^q \right]^{1/q},
}
\]
where 
\[
X_t^p = \exp\left( p\alpha M (t) -\frac{(p \alpha)^2}{2} \left[ M \right]_t\right)
\] 
is an exponential martingale with constant mean equal to $1$. Moreover, 
\begin{eqnarray*}
Y_t^q & = \exp \left( pq \frac{\alpha^2}{2}  \left[ M \right]_t 
+ \alpha q \int_0^{t} ({\mc L}_{N,{\rm open}} H)(\eta (s)) ds \right) \cr
& = \exp\left( \int_{0}^t \Big(  F( \eta_{-N} (s), \lambda_\ell, T_\ell) + 
F( \eta_{N} (s), \lambda_r, T_r) \Big) ds \right),
\end{eqnarray*}
with
\[
F( r , \lambda, T) = \lambda \alpha q \left[  (p\alpha T -1) (V'(r))^2 + T V'' (r) \right].
\]
Taking $\alpha$ sufficiently small, and using~(\ref{eq:potpot0}) and~(\ref{eq:potpot}), 
we see that there exists two constants $C,c>0$ 
(depending on $\alpha,p, T_\ell, T_r, \lambda_\ell, \lambda_r$ and $V$) such that
\[
\forall (r,\lambda,T) \in \mathbb{R}\times\{ \lambda_\ell, \lambda_r \} \times\{ T_\ell, T_r\}
\qquad 
F(r, \lambda,T) \le qc - C[\theta'(r)]^2.
\]
This completes the proof of the lemma.
\end{proof}

\subsection{Smoothness of the transition probability}
\label{sec:smoothness_pbm}

The generator ${\mc L}_{N,{\rm open}}$ can be written as
\[
{\mc L}_{N,{\rm open}} = {X}_0 +\lambda_\ell T_\ell {X}_{-N}^2 + \lambda_r T_r {X}_{N}^2,
\]
where $X_0,X_{\pm N}$ are first-order differential operators: 
\[
{X}_0 ={\mc A}_N - \lambda_\ell V' (\eta_{-N})\partial_{\eta_{-N}} 
- \lambda_r V' (\eta_{N})\partial_{\eta_{N}}, 
\qquad  
{X}_{\pm N} = \partial_{\eta_{\pm N}}.
\]
If the Lie algebra ${\mf L}$ generated by the vector fields $X_0, X_{\pm N}$, 
\textit{i.e.} the smallest Lie algebra containing 
\begin{equation*}
\fl \qquad \{ X_i\}_{ i \in \{-N,0,N\}}, \quad \{[ X_i, X_j]\}_{i,j \in \{-N,0,N\}}, 
  \quad \{[X_i, [ X_j, X_k]\}_{i,j,k \in \{-N,0,N\}}, \quad \ldots
\end{equation*} 
has full rank at every point $\eta$, then the generator ${\mc L}_{N,{\rm open}}$ is
hypoelliptic, and, by H\"ormander's theorem on hypoelliptic operators (see~\cite{H}), 
the semigroup $(T_t)_{t \geq 0}$ 
generated by ${\mc L}_{N,{\rm open}}$ has a smooth transition probability 
density, is strong Feller and the invariant measures, if they exist, 
also have a smooth density.  

To show that ${\mf L}$ has full rank,
we first observe that $\partial_{\eta_N} \in {\mf L}$. 
We then prove that $\partial_{\eta_{N-1}} \in {\mf L}$. Indeed, 
$[X_0, \eta_N]= -V'' (\eta_N) \partial_{\eta_{N-1}} + \lambda_r V'' (\eta_N)\partial_{\eta_N}$,
so that $-V'' (\eta_N) \partial_{\eta_{N-1}} \in {\mf L}$. 
If $V'' (\eta_N) \ne 0$ we already have $\partial_{\eta_{N-1}} \in {\mf L}$. 
Otherwise we compute the iterated Lie bracket 
$[\ldots,[X_{0}, \partial_{\eta_N} ],\ldots,{\partial_{\eta_N}}]$ and obtain 
$V^{(m)} (\eta_N) \partial_{\eta_{N-1}} \in {\mf L}$ for any $m \ge 2$. 
The non-degeneracy condition (\ref{eq:nd}) on $V$ therefore gives 
$\partial_{\eta_{N-1}} \in {\mf L}$. By iterating the above argument, it follows 
easily that $\partial_{{\eta_x}} \in {\mf L}$ for every $x \in \{-N, \ldots,N\}$. 

\subsection{Irreducibility of the dynamics}
\label{sec:uniqueness_pbm}

We show here that the semigroup $(T_t)_{t \ge 0}$ is strongly irreducible, that is, for every $\eta$, every $t>0$, and every non-empty open set $A$, it holds $T_t (\eta, A) >0$. Recall that, in general, irreducibility is not a consequence of hypoellipticity.  

Irreducibility can be proved by using a well-known relationship between stochastic 
differential equations and control theory (see for instance the discussions 
in~\cite[Section~3]{EPRB} or~\cite[Section 4]{C}). In~\cite{Ha}, a general approach for 
obtaining the desired controllability is presented for divergence-free systems having a 
conserved quantity and satisfying a H\"ormander condition. 
Since our system satisfies these assumptions, we can apply Theorem 2.1 of~\cite{Ha} 
and deduce that the semigroup $(T_t)_{t \ge 0}$ is strongly irreducible. 

%--------------- Cas gamma > 0 ---------------------
\section{Existence and uniqueness of the stationary state for $\gamma>0$}
\label{sec:app_B}

We denote by $(T_t)_{t \ge 0}$ the semigroup generated by ${\mc L}_{N,{\rm open}}$ 
for some given $\gamma>0$, and by $({\widetilde T}_t)_{t \ge 0}$ the semigroup 
corresponding to ${\mathcal L}_{N,{\rm open}}$ for~$\gamma=0$. 
The same arguments as for the case $\gamma=0$ can be used to show that 
$(T_t)_{t \ge 0}$ satisfies some Lyapunov condition similar to~(\ref{eq:TTT}).
Relying on~\cite[Theorem~8.1]{RBa} for instance, 
it is then enough to show that the process is irreducible and that the transition 
probability has a smooth density.  

\subsection{Irreducibility of the dynamics}

We show first that $(T_t)_{t \ge 0}$ is strongly irreducible. 
Let us denote the probability transition of ${\widetilde T}_t$ 
by $\widetilde{p}_t (\eta, \xi)\,d\xi$ and the probability 
transition of ${T}_t$ by $p_t (\eta, d\xi)$.

Let $\sigma_1$ be the stopping time defined as the first time when two variables $\eta_x$ 
and $\eta_{x+1}$ are exchanged. Observe that $\sigma_1$ has an exponential law 
of parameter $2\gamma N$. For every bounded measurable function 
$f\,:\,\RR^{\{-N, \ldots,N\}} \to \RR$, it holds
\begin{equation}
\label{eq:sgi}
\fl \qquad \eqalign{
(T_t f) (\eta) &= {\mathbb E}_{\eta} \left[ f(\eta (t)) {\bf 1}_{\sigma_1  \ge t}\right]
+ {\mathbb E}_{\eta} \left[ f(\eta (t)) {\bf 1}_{\sigma_1 <t}\right]\cr
& = e^{-2N \gamma t} \int_{\xi} {\widetilde p}_{t} (\eta, \xi) f(\xi) d\xi \cr
& \quad + {\gamma} \int_0^t ds \, e^{-2\gamma N s}\;\sum_{x=-N}^{N-1}\, \int_{\xi} d\xi\, {\widetilde p}_s (\eta, \xi) \left( \int_{\xi' } p_{t-s} (\xi^{x,x+1}, d \xi') f(\xi') \right).
}
\end{equation}
Iterating the above argument, we obtain the following formula for $p_t$:
\begin{equation*}
\fl \eqalign{
p_t (\eta, d\xi) & = e^{-2\gamma N t} {\widetilde p}_{t} (\eta, \xi) d\xi \cr
& + \sum_{k=1}^{\infty}  {\gamma^k} \sum_{x_1, \ldots, x_k=-N}^{N-1} \left[ \int_0^\infty \ldots \int_0^\infty  ds_1 \ldots ds_{k+1}  e^{-2 \gamma N (s_1 +\ldots +s_{k+1})} {\bf 1}_{\{ s_1 + \ldots +s_k \le t < s_1 + \ldots s_{k+1}\}} \right.\cr
& \left. \qquad \int_{\xi_1, \ldots ,\xi_k}  {\widetilde p}_{s_1} (\eta, \xi_1) {\widetilde p}_{s_2} (\xi^{x_1, x_1 +1}_1, \xi_2) \ldots \right.\cr
& \left. \phantom{\ldots \ldots \ldots } \ldots {\widetilde p}_{s_k} (\xi_{k-1}^{x_{k-1}, x_{k-1} +1}, \xi_k) {\widetilde p}_{t-(s_1 +\ldots s_k)} (\xi_k^{x_k, x_{k} +1} , \xi) d\xi_1 \ldots d\xi_k \right] d\xi.
}
\end{equation*}
This shows that $p_t (\eta, d\xi) = p_t (\eta,\xi)d\xi $ is absolutely continuous with respect to the Lebesgue measure. Moreover, the semigroup $(T_t)_{t \ge 0}$ is strongly irreducible because ${\widetilde T}_t$ is strongly irreducible.

\begin{lemma}
The semigroup $(T_t)_{t \ge 0}$ is strongly Feller, \emph{i.e.} it maps bounded measurable functions to continuous bounded functions.
\end{lemma}

\begin{proof}
The semigroup $({\widetilde T}_t)_{t \ge 0}$ is strongly  Feller. This implies (see \emph{e.g.}~\cite[Corollary~2.4]{SW}) that for every $t>0$ and every compact set $K$,
\begin{equation}
\label{eq:sf}
\lim_{\delta \to 0} \; \sup_{\substack{|\eta - \eta'| \le \delta \cr \eta, \eta' \in K}}\;  
\sup_{\| u \|_{\infty} \le 1} 
\left| \left({\widetilde T}_t u\right)(\eta) -\left({\widetilde T}_t u \right)
(\eta') \right| = 0.
\end{equation}
Let $f$ be a bounded measurable function with $\| f \|_{\infty} \le 1$. We have to show that, for any fixed $t>0$, $T_t f$ is a continuous bounded function. By (\ref{eq:sgi}) we have
\begin{eqnarray*}
&\fl (T_{t} f)(\eta') -(T_t f)(\eta) =e^{-2\gamma N t} \left(  ({\widetilde T}_t f)(\eta')-  ({\widetilde T}_t f)(\eta)  \right)\cr
&\fl \qquad +\gamma \sum_{x=-N}^{N-1} \int_{0}^t  e^{-2 \gamma N (t-s)} \left\{ \left( {\widetilde T}_{t-s} \circ F_x \circ T_s \circ f\right) (\eta) -\left( {\widetilde T}_{t-s} \circ F_x \circ T_s \circ f \right) (\eta')\right\} ds.
\end{eqnarray*}
Here, $F_x$ is the operator acting on functions $f \equiv f (\eta)$ as $(F_x f)(\eta)= f(\eta^{x,x+1})$. Observe that the absolute value of  the second term on the right hand side is bounded above by 
\begin{equation*}
\gamma \sum_{x=-N}^N \int_0^t e^{-2N\gamma (t-s)} \; \sup_{\| g \|_{\infty} \le 1} \left| \left( {\widetilde T}_{t-s} g\right) (\eta) - \left( {\widetilde T}_{t-s} g\right) (\eta')  \right| ds.
\end{equation*} 
By the bounded convergence theorem and (\ref{eq:sf}) we have 
\begin{equation*}
\lim_{\eta' \to \eta} (T_t f)(\eta') - (T_t f)(\eta) = 0, 
\end{equation*}
which concludes the proof.
\end{proof}

These two last properties (irreducibility and strong Feller property) are sufficient to have uniqueness of the invariant measure $\mu_{\rm ss}$. To show that the latter has a density, we observe that for any $t>0$, the condition $\mu_{\rm ss} T_t =T_{t}$ implies that, for any measurable set $A$ of $\RR^{\{-N, \ldots,N\}}$,
\begin{eqnarray*}
  \mu_{\rm ss} (A) 
  & = \int  d\mu_{\rm ss} (\eta) \left( \int {\bf 1}_A (\xi) p_{t} 
  (\eta, \xi) d\xi \right)\cr
  & = \int  {\bf 1}_A (\xi) \left( \int d\mu_{\rm ss} (\eta) p_{t} (\eta, \xi) \right) d\xi,
\end{eqnarray*}
where the second line follows from Fubini's theorem.

\ack
We thank the referees for their stimulating suggestions,
as well as J. Fritz for very relevant technical comments.
We acknowledge the support of the French Ministry of 
Education through the grants ANR-10-BLAN 0108 (SHEPI) and ANR-09-BLAN-0216-01 (MEGAS).

%-------------- BIBLIO ------------------------

\end{document}